\documentclass[a4paper,reqno]{amsart}
\usepackage[english]{babel}
\usepackage[T1]{fontenc}
\usepackage{amsmath,amssymb,amsfonts,amsthm}
\usepackage{hyperref}
\usepackage{slashed}
\usepackage{graphicx,color}
\usepackage[bbgreekl]{mathbbol}
\usepackage{mathtools}
\usepackage[all,cmtip]{xy}
\usepackage[normalem]{ulem}
\usepackage{todonotes}

\theoremstyle{definition}
\newtheorem{definition}{Definition}
\newtheorem{notation}[definition]{Notation}
\newtheorem{theorem}[definition]{Theorem}
\newtheorem{proposition}[definition]{Proposition}
\newtheorem{propdef}[definition]{Proposition/Definition}
\newtheorem{lemma}[definition]{Lemma}
\newtheorem{example}[definition]{Example}

\newtheorem{remark}[definition]{Remark}

\newcommand{\g}{\mathfrak{g}}
\newcommand{\calA}{\mathcal{A}}

\newcommand{\de}{\partial}

\newcommand{\qsp}[2]{\,\ensuremath{\raise.5ex\hbox{$#1$}\big\slash\raise-.5ex\hbox{$#2$}}}

\newcommand{\RR}{\mathbb{R}}
\newcommand{\intl}{\int\limits}

\newcommand{\cA}{\mathcal{A}}
\newcommand{\cB}{\mathcal{B}}

\newcommand{\LQ}{\mathcal{L}_{Q}}
\newcommand{\LE}{\mathcal{L}_{\mathcal{E}}}
\newcommand{\bD}{\mathbb{\Delta}}
\newcommand{\D}{\mathbb{D}}
\newcommand{\del}{\partial}
\newcommand{\delbar}{\bar{\partial}}
\newcommand{\oloc}{\Omega_{\mathrm{loc}}}
\newcommand{\xev}{\mathfrak{X}_{\mathrm{evo}}}
\newcommand{\T}{\mathbb{T}}
\newcommand{\mr}{\mathrm}

\newcommand{\dd}{\partial}
\renewcommand{\O}{\mathcal{O}}
\newcommand{\BBX}{\mathbb{X}}

\title{Towards Holography in the BV-BFV setting}
\author{Pavel Mnev}
\address[P. Mnev]{Department of Mathematics, Notre Dame University, 255 Hurley Bldg, Notre Dame, IN  46556-4618, U.S.A.}
\address[P. Mnev]{St. Petersburg Department of V. A. Steklov Institute of Mathematics of the Russian Academy of Sciences, Fontanka 27, St. Petersburg, 191023 Russia}
\email{pmnev@nd.edu}

\author{Michele Schiavina}
\address[M. Schiavina]{Institute for Theoretical Physics, ETH Z\"urich, Wolfgang Pauli strasse 27, 8093, Z\"urich, Switzerland}
\address[M. Schiavina]{Department of Mathematics, ETH Z\"urich, R\"amistrasse 101, 8092, Z\"urich, Switzerland}
\email{micschia@phys.ethz.ch}

\author{Konstantin Wernli}
\address[K. Wernli]{Department of Mathematics, Notre Dame University, 255 Hurley Bldg, Notre Dame, IN  46556-4618, U.S.A.}

\email{kwernli@nd.edu}
\thanks{This research was (partly) supported by the NCCR SwissMAP, funded by the Swiss National Science Foundation, and by the COST Action MP1405 QSPACE, supported by COST (European Cooperation in Science and Technology). P. M. acknowledges partial support of RFBR Grant No. 17-01-00283a. M.S. acknowledges partial support by SNF grant No. P300P2\_177862. K. W. acknowledges partial support of SNF Grant No. 200020 172498/1, a Forschungskredit of the University of Zurich, grant No. FK-16-093, GRC Travel Grant 2017\_Q3\_TG\_005, a Dirichlet Fellowship of the Berlin Mathematical School, and SNF Postdoc fellowship P2ZHP2\_184083. Parts of this research were completed while K. W. was affiliated with the University of Zurich and Humboldt University, he would like to thank them for providing excellent work environments. P. M. would like to thank Andrey S. Losev and Nicholas J. Teh for inspiring discussions. M. S. would like to thank the University of Notre Dame for facilitating collaboration on this project. The authors would like to thank Alberto S. Cattaneo for helpful discussions and the anonymous referee for helping improve the paper.}

\makeatletter
\providecommand\@dotsep{5}
\def\listtodoname{List of Todos}
\def\listoftodos{\@starttoc{tdo}\listtodoname}
\makeatother

\begin{document}

\maketitle

\begin{abstract}
We show how the BV-BFV formalism provides natural solutions to descent equations, and discuss how it relates to the emergence of \emph{holographic counterparts} of given gauge theories. Furthermore, by means of an AKSZ-type construction we reproduce the Chern--Simons to Wess--Zumino--Witten correspondence from infinitesimal local data, and show an analogous correspondence for BF theory. We discuss how holographic correspondences relate to choices of polarisation relevant for quantisation, proposing a semi-classical interpretation of the quantum holographic principle.
\end{abstract}

\tableofcontents

\section*{Introduction}
A framework to treat classical and quantum field theories on manifolds with boundaries and corners was introduced in a systematic way in \cite{CMR}, joining the seminal works of Batalin, Fradkin and Vilkovisky \cite{BV77,BV81,BF83}, by establishing a correspondence between data associated to a field theory on a \emph{bulk} manifold $M$ and data associated to its boundary and possibly corners. The (semi-)classical part of the formalism produces a \emph{resolution} of the space of classical solutions to a given variational problem modulo gauge transformations - the BV complex - together with a cohomological description of its \emph{Hamiltonian} structure\footnote{This is sometimes called the \emph{reduced phase space}.} - the BFV complex - and a correspondence between the two (a fibration). This data constitutes a classical BV-BFV pair, which is then used as an input for a perturbative quantisation scheme that is, by construction, compatible with gluing of manifolds along common boundaries \cite{CMR2}.

The aim of this paper is to argue how the BV-BFV approach to field theory on manifolds with boundaries and corners offers a natural framework to understand the emergence of \emph{edge modes}  - degrees of freedom supported on higher codimension strata - and their relation to their \emph{parent} field theory in the bulk, a correspondence that might be interpeted as a semi-classical analogue of holography. Very famous instances of this correspondence, such as the one between Chern--Simons and Wess--Zumino--Witten theories \cite{WZ,W83,W84, W89, GK, EMSS} or general relativity and conformal field theories (e.g. in three spacetime dimensions with Liouville theory \cite{CHvD,Carlip}), have been studied extensively and conjectured to hold in larger generality. Similar observations have been made in \cite{DoFr, Gei}, where the notion of edge mode is linked to failure of gauge invariance of the various data that define a theory. However, a full description of the mechanism at work is not yet available, despite a host of physical applications and experiments inspiring the investigation, and providing real-life incarnations of such phenomena. This is especially visible in condensed matter physics, where topologically protected states of matter provide an example of such correspondence, and where edge modes generate measurable quantities such as edge electronic currents (see \cite{CM} and {  references therein).

%By rephrasing the constructions of  to involve local functionals with values in inhomogeneous forms (the variational bicomplex), 

One of the main achievements here is the development of an inhomogeneous form-valued BV-BFV formalism\footnote{One can think of Lagrangian densities as top form-valued functionals of field configurations.}, which is designed to handle field theory with gluing and cutting in the presence of corners \cite{CMR,IM}, and to handle defects. The notion of { \emph{descent equations} (Section \ref{Sec. WDE}, \cite{Z1, MSZ, W1}}) is extended to the BV formalism, it is interpreted as cocycle conditions for the BV-BFV complex (Definition \ref{Def:BV-BFVcomplex}), and it is related to known results for BRST \cite{BRS1,BRS2,BRS3,T} and group-cohomology constructions \cite{Alekseev}. 

Our construction extracts a universal solution of the descent equations from the data of an \emph{$n$-extended BV-BFV theory}, i.e. a field theory for which the BV-BFV axioms hold - recursively - up to codimension $n$ (Definitions \ref{laxBVBFV} and \ref{strictification}). We will call this solution the \emph{total Lagrangian} of the theory, and denote it by $\mathbb{L}^\bullet$ (cf. Theorem \ref{Deltacocycle}). Furthermore, we show how possible alternative choices of polarisations at the prequantum level lead to other solutions.

We argue that promoting a BV functional (i.e. a solution to the classical master equation up to boundary terms) to an object { satisfying descent}, is indeed closely related to the emergence of edge modes.

In this setting, we propose a new program to approach the study of holography, starting from an analysis that is - by construction - compatible with the modern mathematical quantisations schemes of \cite{CMR2} and, to some extent, of \cite{Co,CG} and \cite{Re, BFR}.

The holographic construction we propose, in our interpretation, changes the type of any codimension-$1$ stratum from ``Segal-type" (i.e. along which one can cut/glue) to ``holographic-type", i.e. along which one can no longer cut/glue, but which carries degrees of freedom and an emerging action functional that corrects the gauge anomaly of the bulk action functional.

Considerations on the quantisation of the theories considered here - aside from a few comments in Sections \ref{Sec:Polarisationcomments} and \ref{Sect:Omegacohomology} - will be postponed and explored in a subsequent paper. We remark, however, that our construction controls prequantum data by linking cocycles of an appropriate complex to choices of polarisations.

\subsection*{Summary of Results}
The key observation in this paper is that, at every codimension, there exist two naturally induced functionals $L^\bullet$ and $L^\bullet_{CMR}$. The former represents the failure of the BV classical master equation in the presence of higher strata, while the latter encodes the failure of BV-invariance. Their difference $\bD^\bullet$ (Definition \ref{BVBFVDIFFERENCES}) turns out to be of central relevance. 

First, we show that $\bD^\bullet$ satisfies the BV-descent equations (neither $L^\bullet$ or $L^\bullet_{CMR}$ normally do). Then, if a theory is constructed out of BRST data, we show how $\bD^\bullet$ reduces to a solution of  BRST-descent equations (Section \ref{Sec:BV/BRST} and Theorem \ref{BRSTtypeTheorem}). This is interpreted as an appropriate choice of \emph{polarisation} in the associated symplectic spaces at every codimension. This \emph{freedom of choice} comes from a structural symmetry of the BV-BFV equations, leading to the notion of an $f$-transformation (see Definition \ref{polarisationdef} and the discussion in Remark \ref{polarisationremark}).

To see why this relates to holography, we implement a construction that stems from the AKSZ approach to field theory \cite{AKSZ} (Section \ref{Sec:AKSZ}), to perform integration of Lie algebra-valued fields to Lie group-valued ones. We consider the diagram
\begin{equation}
	\xymatrix{
		\mathrm{Map}(T[1]I, \mathcal{F}^{(1)})\times T[1]I \ar[d]_{p}   \ar[r]_-{\mathrm{ev}}  & \mathcal{F}^{(1)}\\
		\mathrm{Map}(T[1]I, \mathcal{F}^{(1)}) & 
	},
\end{equation}
where $\mathcal{F}^{(1)}$ is a space of codimension-$1$ fields for a given (strict) BV-BFV theory, and the \emph{transgression} map 
$$
	\T^\bullet_I\colon\Omega^\bullet\left(\mathcal{F}^{(1)}\right)
		\longrightarrow \Omega^\bullet\left(\mathrm{Map}(T[1]I, \mathcal{F}^{(1)})\right),
$$ 
given by the composition $\T^\bullet_I:=p_* \mathrm{ev}^*$. 

In Section \ref{Sec:CS-WZW}, looking at the guiding example of Chern--Simons theory on a three-dimensional manifold with boundary $(M,\partial M)$, we map $\bD^\bullet_{CS}$ to \emph{gauged Wess--Zumino} (gWZ) and \emph{gauged Wess--Zumino--Witten} (gWZW) functionals \cite{GK} (see Definition \ref{WZWdef} to fix the notation). We call this procedure AKSZ integration.

Indeed, by solving certain natural equations inside the space of AKSZ fields\footnote{In particular one considers the part of the EL locus for expressions that have one-form component along $I$, in degree zero. These can also be interpreted as evolution equations for degree-zero maps. We denote such critical fields by $\mathrm{dgMap}^0_I(T[1]I,\mathcal{F}^{(1)}_{CS})$.} one finds a surjective map $\mathcal{I}$ onto the space of Wess--Zumino fields, and obtains (Theorem \ref{Theorem:WZWfromDelta})

\begin{equation}
	\left[\T_{I}^0 \D_{f_{min}}^{(1)}\right]_{\mathrm{dgMap}_I^0} = \mathcal{I}^*S_{gWZ},
\end{equation}
where $\D^{(1)}_{f_{min}}$ is cohomologous to $\int_{\partial M}\bD^{\bullet}_{CS}$.  This, in words, means that the gauged Wess--Zumino functional is the AKSZ integration of the difference $\bD^\bullet_{CS}$ in the polarisation induced by the functional $f_{min}$.}

To recover the \emph{kinetic} part of the celebrated gauged Wess--Zumino--Witten functional, we choose a conformal structure on $\partial M$ (hence inducing a different polarisation $\mathcal{P}$ in $\mathcal{F}^{(1)}_{CS}$) and, by changing the data consistently, we are able to show that
\begin{equation}
	\left[\T_{I}^0 \D^{(1)}_{f_{min}^{1,0}}\right]_{\mathrm{dgMap}_I^0}=\mathcal{I}^*S^{1,0}_{gWZW}.
\end{equation}

{ 
Phrased in this language, the only difference between gWZ and gWZW theories - both obtained from $\bD_{CS}^\bullet$ via AKSZ integration -} is the choice of a particular representative in the cohomology class of $\bD_{CS}^\bullet$, a choice that might depend on a complex structure (or metric)\footnote{In fact this is necessary to define the WZW functional.}, and which relates to a choice of polarisation on the space of boundary fields for Chern--Simons theory. 

{  We interpret this AKSZ construction as adjoining a partly \emph{on-shell} collar to the manifold with boundary}. This process is supposed to modify the state associated to the bulk (after quantisation) by a multiplicative factor that takes into account the choice of a polarisation, making the resulting state manifestly gauge invariant. Following this interpretation, the gauged Wess--Zumino--Witten partition function becomes the \emph{effective} result of gluing to the boundary of Chern--Simons an \emph{AKSZ field theory} supported on a cylinder, with target functional given by the BV-BFV difference $\bD^\bullet_{CS}$.

In Section \ref{BFTheory}, analogous results are obtained in the case of three dimensional BF theory, where the AKSZ integration of the appropriate representative of $\bD^{\bullet}_{BF}$ recovers the failure of the gauge invariance of the classical BF action functional (Definition \ref{propdefBFWZW}). Although historically the failure under gauge invariance of the BF action functional has received less attention than Chern--Simons functional, it appears to be conceptually analogous.

{ By interpreting finite gauge transformations for BF theory as the action of the \emph{double} Lie group $\widetilde{G}=G\ltimes \mathfrak{g}^*$, we construct a gauged Wess--Zumino type functional\footnote{Observe that such functional is manifestly a boundary term.}
\begin{equation}
	S_{\tau F}[g,\tau ,A] = \intl_{\partial M} \tau^{g^{-1}} F_A
\end{equation} 
with the property that $S_{BF}^{cl}[(A,B)^{(g,\tau)}] - S_{BF}^{cl}[(A,B)] = S_{\tau F}[g,\tau ,A]$. Applying the AKSZ integration, we show that the difference at codimension-$1$ for BF theory, $\D^{(1)}_{BF}$, correctly encodes such failure:\footnote{Again $\D_{f_{min}}$ is cohomologous to $\D_{BF}^{(1)}$, with the latter being identically zero for BF theory.}
\begin{equation}
	\left[\T^0_I \D^{(1)}_{f_{min}}\right]_{\mathrm{dgMap}^0_I}=\mathcal{I}^*S_{\tau F}.
\end{equation}
In this context, we also show how BF theory - as a fully extended BV-BFV theory - can be seen as the result of a particular $f$-transformation of Chern--Simons' BV-BFV data, with structure group the double Lie group $\widetilde{G}$.

In Section \ref{YMTheory} we discuss certain aspects of Yang--Mills theory in the BV-BFV formalism and highlight the particular behaviour of the functional $\bD_{YM}^\bullet$ in this scenario, relating to known work on edge modes by Donnelly--Freidel \cite{DoFr} (see Remark \ref{YMRemark}).

Finally, in Section \ref{Sect:PSM} we review the BV-BFV construction for the Poisson Sigma model (PSM), presenting two different f-transformations and showing that when the $1$-stratum is of holographic type, one recovers a version of topological classical mechanics, as holographic counterpart.

\subsection*{Outlook and extensions of this work}
{  The AKSZ integration procedure presented here extends the idea of integrating Lie algebras to Lie groups by means of paths of flat connections to general field theories, possibly not of AKSZ-type themselves, by phrasing the construction in terms of differential graded maps \cite{CSW}.}

We would like to stress that, despite the procedure being tested on and inspired by known results on the classical Chern--Simons/Wess--Zumino--Witten correspondence, it has shown to be predictive enough to allow us to \emph{deduce} the example of non-abelian BF theory and the Poisson Sigma model.  This, in addition, embeds said results in a rigorous, covariant perturbative quantisation scheme with corners \cite{IM}.

We are positive that the presented mechanism can have a strong predictive power in more involved examples, including general relativity (GR), yielding \emph{holographic counterparts} from a pre-quantum approach to field theory. 

We defer a full analysis of GR and PSM to a later paper. We note, however, that the semiclassical holographic counterpart for GR in three dimensions follows from the results presented in this paper, combined with the (strong) BV-equivalence betewen GR in triad formalism and BF/Chern--Simons theory (see \cite{CSS,CaS} for the explicit equivalence at all codimensions and \cite{Carlip} for the link to Liouville theory.). It should be remarked that 4d-GR has proven to be somewhat difficult to extend to higher codimensions in a strict sense. An obstruction\footnote{The problem arises when trying to define the space of codimension-$1$ fields as pre-symplectic reduction of the natural space of fields induced on a codimension-$1$ stratum, hence it is a strictification problem.} for tetradic GR was found in \cite{CSPCH}, while standard metric gravity was shown to be (at least) strictly 1-extendable in \cite{CSEH}. We believe that the lax approach to the BV-BFV axioms presented here might help overcome (or bypass) such obstructions.

Among the goals of this paper is also to show the advantage of using BV over the BRST formalism, even for theories that can otherwise be treated with \emph{standard} methods. BV-BFV provides an overarching framework for phenomena that might go beyond theories with symmetries that close off-shell, and it allows direct access to higher codimension structures that would be more cumbersome to compute with more traditional approaches, and less clear from a conceptual point of view.

Finally, we would like to remark that, although our implementation of the AKSZ transgression procedure currently stops at codimension-$1$, there is no conceptual obstruction to investigating higher codimension factors of $\bD^\bullet$. We believe these observations are useful to better understand defects and anomalies in quantum field theory, and that they already shed light on the nontrivial phenomenon of holography in classical and quantum field theory. As a matter of fact, in Yang--Mills theory one can see how the choice of a holographic-type BV-BFV data on a codimension-$2$ stratum (a corner) can explain the gauge anomaly in the codimension-$1$ symplelctic structure, partly recovering some observations by Donnelly and Freidel \cite{DoFr}. A similar analysis of the edge-modes phenomenon is being carried out by \cite{ND}.

}

\section{Framework}
In this paper we will be concerned with a density version of the classical BV-BFV approach to field theories, as presented in \cite{CMR}. In order to accommodate this, we will present a relaxed version of the BV-BFV axioms, that will allow us to deduce some general results on the algebraic structure underlying such axioms (Definition \ref{laxBVBFV}). 

In the applications, a stricter notion will be needed for several purposes, among which we mention the possibility of finely distinguishing between theories, and the access to a (pre-)quantisation scheme \cite{CMR2}. This \emph{strictification} will be presented in Definition \ref{maindef}.

What will be called a \emph{space of fields} $\mathcal{F}$ should be thought of as the space of smooth sections of a graded vector bundle (or sheaf) over an $m$-dimensional spacetime manifold $M$. In the applications, one wants to allow $M$ to have boundary and corners, and, more generally, carry a stratification. 

\begin{definition}
We define a stratification of a manifold $M$ to be a filtration of $M$ by CW-complexes $\{M^{(k)}\}_{k=0\dots m}$ such that, for each $k$, $M^{(k)}\backslash M^{(k+1)}$ is a smooth $(m-k)$-dimensional manifold. Its connected components are the codimension-$k$ strata. In what follows, manifolds will always be assumed to be oriented.
\end{definition}

Throughout, we will consider \emph{local} functionals and \emph{local} forms with values in inhomogeneous differential forms on $M$ (see Definition \ref{localforms} below, following \cite{Anderson, Delgado, DF}). This will enable us, upon specifying a stratification in $M$, to integrate the aforementioned forms and obtain the usual (strict) BV-BFV data. {  Another description of local forms can be found in \cite{BBH}.}

\begin{definition}[Local forms and integrated local forms]\label{localforms}
A \emph{local} form on a (possibly graded) vector bundle $E\longrightarrow M$ on an $m$-dimensional manifold $M$, is an element of 
\begin{equation}
	\left(\oloc^{\bullet,\bullet} (\mathcal{F}\times M),\delta,d\right) := 
		(j^\infty)^*\left(\Omega^{\bullet,\bullet}(J^\infty (E)),d_V,d_H\right)
\end{equation}
where $\mathcal{F}:=\Gamma^\infty(M,E)$, $j^\infty$ is the limit of the maps $\{j^p\colon \mathcal{F}\times M \longrightarrow J^pE\}$ with $J^pE$ the $p$-th jet bundle of $E$,  $J^\infty$ is the limit of the sequence 
$$
	E\equiv J^0E \longleftarrow J^1E \longleftarrow \dots \longleftarrow J^p E \longleftarrow \dots
$$
and $\oloc^{\bullet,\bullet}(\mathcal{F},\times M)$ is endowed with the differentials
\begin{align}
	\delta (j^\infty)^*\alpha & \coloneqq (j^\infty)^* d_V\alpha\\
	d (j^\infty)^*\alpha & \coloneqq (j^\infty)^* d_H\alpha.
\end{align} 
An element of $\oloc^{0,\bullet}(\mathcal{F}\times M)$ is called \emph{Local Functional}. 

An \emph{integrated local form} on $E\longrightarrow M$ is the integral along an $(m-k)$ (sub)manifold $M^{(k)}\longrightarrow M$ of an element of $\Omega^{\bullet,m-k}(\mathcal{F}\times M)$. We will denote the complex of integrated local forms by $\left(\oloc^\bullet(\mathcal{F},M), \delta\right)$.
\end{definition}

We will not provide a full exposition on the theory of local vector fields on $E$ (presented e.g. in \cite{Anderson,Delgado,DF}), as we will only need the following notion.

\begin{definition}
An \emph{evolutionary vector field} $X\in\xev(\mathcal{F})$ on $\mathcal{F}$ is a vector field on $J^\infty E$ which is vertical with respect to the projection $J^\infty E\longrightarrow M$, and such that 
\begin{equation}\label{evolutionary}
	\mathcal{L}_X d = d \mathcal{L}_X,
\end{equation}
where $\mathcal{L}_X=[\iota_X,\delta]$ is the variational Lie derivative on local forms.
\end{definition}

\begin{remark}
The spaces of fields we will consider throughout are usually thought of as (tame) Frech\'et spaces, and smoothness of maps is generally regarded in the Frech\'et sense. However, depending on the kind of statements one is after, other types of topologies might be better suitable. Since, for the purposes of this paper we are content with Cartan calculus and standard differential geometry on local objects, we shall not distinguish such topologies. We refer to \cite{Delgado} for a modern review on the issue of smoothness on spaces of fields.
\end{remark}

\subsection{BV-BFV formalism}
In this section we will work with local functionals and one-forms on graded vector bundles, and we will focus on the interplay between two independent grading: the $M$-form degree and the internal grading in $E$, called \emph{ghost number}.

\begin{definition}
The internal grading of the vector bundle $E\longrightarrow M$, inherited by its sections and all polynomial functions on them, is called \emph{ghost number}, and is denoted by $\mathrm{gh}$. Local forms on the vector bundle $E$ have a natural horizontal co-form-degree denoted by $\#$, computed as $\mathrm{dim}(M)$ minus the form-degree on $M$, and the additional ghost number. We will define the total degree to be the difference: $\mathrm{deg}= \mathrm{gh} - \#$. 
\end{definition}

\begin{definition}[\footnote{This definition was in part inspired by a private communication of P.M. with E. Getzler, ca. 2014.}]\label{laxBVBFV}
A \emph{lax BV-BFV theory} is the assignment, to a manifold $M$, of the data 
\begin{equation}
	\mathfrak{F}=(\mathcal{F}, L^\bullet, \theta^\bullet, Q)
\end{equation}
with 
\begin{itemize}
\item $\mathcal{F}$ the space of $C^{\infty}$ sections of a graded bundle (or sheaf) $E\longrightarrow M$, 
\item $L^\bullet\in\oloc^{0,\bullet}(\mathcal{F}\times M)$, a local functional of total degree $0$ 
\item $\theta^\bullet\in\oloc^{1,\bullet}(\mathcal{F}\times M)$, a local one-form of total degree $-1$,
\item $Q\in \xev(\mathcal{F}\times M)[1]$, a degree-$1$, evolutionary, cohomological vector field on $\mathcal{F}$, i.e. $[\mathcal{L}_Q,d]=[Q,Q]=0$, 
\end{itemize}
such that
\begin{subequations}\label{CMReqts}\begin{align}\label{BVBFVrel}
		&  \iota_{Q} \varpi^{\bullet} = \delta L^{\bullet} + d\theta^\bullet\\\label{bracket}
		&\frac12 \iota_Q\iota_Q\varpi^{\bullet} = d L^{\bullet},
\end{align}\end{subequations}
where $\varpi^\bullet\coloneqq\delta \theta^\bullet$.

\end{definition}

\begin{remark}
Notice that Equations \eqref{CMReqts} are invariant under $\theta \mapsto \theta + \delta f, L\mapsto L + df$. We call this an $f$-transformation. It will play a central role in this paper, and we will discuss it in more detail in Remark \ref{polremark} and Definition \ref{polarisationdef}.
\end{remark}

\begin{definition}\label{maindef}
A \emph{strict, $n$-extended, exact BV-BFV theory}, shorthanded with \emph{$n$-extended theory}, is the assignment, to the $m$-dimensional stratification $\{M^{(k)}\}_{k=0\dots m}$ ($m\geq n$), of the data 
	$$\mathfrak{F}^{\uparrow n}=(\mathcal{F}^{(k)}, S^{(k)}, \alpha^{(k)}, Q^{(k)},\pi_{(k)})_{k=0\dots n},$$
such that, for every $k\leq n$,
\begin{enumerate}
\item\label{symplecticcondition} $\mathcal{F}^{(k)}$ is the space of sections of a graded vector bundle $E^{(k)}\longrightarrow M^{(k)}$ and $\alpha^{(k)}\in\oloc^{1}(\mathcal{F}^{(k)},M^{(0)})$ is a degree-$k$ integrated local form, such that $\Omega^{(k)}=\delta\alpha^{(k)}$ is weakly symplectic on $\mathcal{F}^{(k)}$,
\item $\pi_{(k)}\colon \mathcal{F}^{(k)} \longrightarrow \mathcal{F}^{(k+1)}$ is a degree-$0$ surjective submersion,
\item $Q^{(k)}$ is a degree-1, evolutionary, cohomological vector field on $\mathcal{F}^{(k)}$, i.e. $[\mathcal{L}_{Q^{(k)}},d]=[Q^{(k)},Q^{(k)}]=0$, that is also projectable: $Q^{(k+1)} = (\pi_{(k)})_* Q^{(k)}$,
\item $S^{(k)}\in \Omega_{\mathrm{loc}}^{0}(\mathcal{F}^{(k)},M^{(k)})$ is a (real-valued) degree-$k$ integrated local functional,
\end{enumerate}
such that
\begin{subequations}\begin{align}
		&  \iota_{Q^{(k)}} \Omega^{(k)} = \delta S^{(k)} + \pi_{(k)}^*\alpha^{(k + 1)}\\
		&\frac12 \iota_{Q^{(k)}}\iota_{Q^{(k)}}\Omega^{(k)} = \pi_{(k)}^* S^{(k+1)},
\end{align}\end{subequations}
and, for $n<k\leq m$, we require $\alpha^{(k)}=S^{(k)}=0$. When $n=m$ we say that the theory is \emph{fully extended}. When n=0, the data defines a \emph{BV theory}.
\end{definition}
\begin{notation}
We will sometimes call $\mathcal{F}^{(k)}$ the \emph{space of fields in codimension-$k$} or \emph{space of codimension-$k$ fields}.  We use codimension to enumerate, as it makes the notation less cumbersome, and because the ghost number  coincides with the co-form/codimension degree. The notation for the inhomogeneous functionals $L^\bullet$ and the integrated forms $S^{(k)}$ follows the usual standard for Lagrangians and action functionals, and we distinguish the inhomogeoneous forms $\theta^\bullet$ and $\varpi^\bullet$ from their integrated versions $\alpha^{(k)}$ and $\Omega^{(k)}$. We will use square brackets $[L]^{m-k}$ to denote the $(m-k)$-form part of the inhomogeneous form $L^\bullet$.
\end{notation}

\begin{remark}
Observe that Condition \ref{symplecticcondition} in Definition \ref{maindef} requires the space $\mathcal{F}^{(k)}$ to be a smooth symplectic manifold, currently ruling out certain versions of general relativity and reparameterisation-invariant models \cite{CSPCH,CStime} but including others \cite{CSEH,CStime,CaS}.
\end{remark}

\begin{remark}
In practical situations Q is defined on \emph{fields}, and extended by prolongation to the jet bundle as an evolutionary vector field.
\end{remark}

\begin{definition}\label{strictification}
An \emph{$n$-strictification} of a lax BV-BFV theory $\mathfrak{F}=(\mathcal{F},L^\bullet, \theta^\bullet,Q)$ on a manifold $M$, is a pairing with an $n$-stratification $\{M^{(k)}\}_{k=0\dots n}$ to yield an $n$-extended, exact BV-BFV theory $\mathfrak{F}^{\uparrow n}=(\mathcal{F}^{(k)}, S^{(k)}, \alpha^{(k)}, Q^{(k)},\pi_{(k)})_{k=0\dots n}$ for which there are surjective submersions 
$$p_{(k)}\colon \mathcal{F} \longrightarrow \mathcal{F}^{(k)}$$
such that 
\begin{enumerate}
\item $p_{(k+1)}=\pi_{(k)}\circ p_{(k)}$, 
\item $Q^{(k)} = (p_{(k)})_* Q$, 
\item $\intl_{M^{(k)}}[\theta^\bullet]^{m-k}=p_{(k)}^*\alpha^{(k)}$,
\item $\intl_{M^{(k)}} [L^\bullet]^{m-k} = p_{(k)}^*L^{(k)}$.
\end{enumerate}
\end{definition}

\begin{remark}
We can deduce from Definition \ref{strictification} that strictifying a lax BV-BFV theory essentially means finding relations between the space of fields in the \emph{bulk} and the spaces of fields at the various strata. Then, the inhomogeneous forms $\theta^\bullet$, and $L^\bullet$ can be integrated along the strata. Observe that often the maps $p_{(k)}$ turn out to be just restriction of fields (and jets), but there are many examples in which this is not the case, and additional reduction is needed, which may also fail to produce the correct data on the $k$-stratum (most notably \cite{CSPCH}).
\end{remark}

\begin{definition}
We define the \emph{Modified Lagrangian}  to be the local functional\footnote{We denote the Modified Lagrangian by $L_{CMR}$ as a reference to Cattaneo, Mn\"ev and Reshetikhin, who introduced (the strict version of) Eq. \eqref{mCME} under the name of Modified Classical Master Equation.} 
\begin{equation}
	L^{\bullet}_{CMR}\coloneqq \left(2L^{\bullet} - \iota_Q\theta^\bullet\right)
\end{equation}
\end{definition}

\begin{lemma}\label{otherrelations}
Let $\mathfrak{F}$ be a lax BV-BFV theory on $M$. Then the following relations hold: 
\begin{subequations}\begin{align}\label{mCME}
	\LQ L^{\bullet}& = dL_{CMR}^\bullet\\ \label{omegacocycle}
	\LQ \varpi^\bullet & = d\varpi^\bullet 
\end{align}\end{subequations}
\end{lemma}

\begin{proof}
This follows from equations \eqref{CMReqts} by contraction with $\iota_Q$ and application of $\delta$, respectively, and recalling that the graded commutators $[Q,d]=[\LQ,\delta]=[\delta,d]=0$ vanish.
\end{proof} 

\begin{remark}\label{Q-dfail}
Because $Q$ also encodes the symmetry data of the theory, Equation \eqref{mCME} morally measures the failure of gauge invariance in the presence of higher codimension strata. This statement becomes precise after the strictification of the lax theory. Observe that Eq. \eqref{omegacocycle} means that $\varpi^\bullet$ is an $(\LQ - d)$-cocycle, whereas Eq. \eqref{mCME} tells us that, generally speaking, $L^\bullet$ is not. 
\end{remark}

\begin{definition}
We define the \emph{graded Euler vector field} $\mathcal{E}$ to be the degree-$0$ evolutionary vector field that acts on ghost-number-homogeneous local forms by 
\begin{equation}
\mathcal{L}_{\mathcal{E}} F = \mathrm{gh}(F) F.
\end{equation}
\end{definition}

\begin{lemma}\label{QEuler}
Let $Q\in\xev(\mathcal{F})$ be a degree 1 evolutionary vector field, then
\begin{equation}
\LQ = \mathcal{L}_{[\mathcal{E},Q]}=\LE\LQ - \mathcal{L}_Q\LE
\end{equation}
and $\LQ\LE = (\LE - 1) \LQ$.
\end{lemma}

\begin{proof}
This is an immediate consequence of Cartan's rule $\mathcal{L}_{[X,Y]} = [\mathcal{L}_X,\mathcal{L}_Y]$, where the bracket is intended as graded, in case the vector fields $X,Y$ have non-zero degree. Moreover, $[\mathcal{E},Q] = \LE Q = Q$ as $Q$ is homogeneous of ghost-number 1. 
\end{proof}

\begin{lemma}
Let $Q$ be a cohomological, evolutionary vector field of degree $1$ on the space of sections $\mathcal{F}$ of the vector bundle $E\longrightarrow M$. Then, the space of local forms $\Omega^{\bullet,\bullet}_{\mathrm{loc}}(\mathcal{F}\times M)$ is a complex with differential given by the Lie derivative $\LQ$. 
\end{lemma}

\begin{proof}
We consider the map $\LQ\colon \oloc^{\bullet,\bullet}(\mathcal{F}\times M)\longrightarrow \oloc^{\bullet,\bullet}(\mathcal{F}\times M)$, which squares to zero due to $[Q,Q]=0$, and raises the ghost number of an inhomogeneous local form by $1$. $\LQ$ is then a differential on $\oloc^{\bullet,\bullet}(\mathcal{F}\times M)$, moreover, from the evolutionary condition $[\LQ,d]=0$ and the standard rules of Cartan calculus $[\LQ,\delta]=0$ we also conclude that it is compatible with both $\delta$ and $d$.
\end{proof}

\begin{definition}\label{Def:BV-BFVcomplex}
Let $\mathfrak{F}$ be a lax BV-BFV theory. We define the \emph{BV-BFV complex} to be the space of local forms with values in inhomogeneous forms on $M$, endowed with the combined differential $\LQ - d$:
\begin{equation}
	\mathbb{\Omega}^{\bullet}(\LQ - d)\equiv\Omega_{\mathrm{BV-BFV}}^{\bullet}
			(\mathcal{F}\times M, \LQ - d)\coloneqq 
		\left(\left(\bigoplus_k	\Omega_{\mathrm{loc}}^{\bullet,k}
			(\mathcal{F}\times M)\right),\LQ-d\right)
\end{equation}
We will also use the shorthand notation $H^\bullet(\LQ - d)$ to denote the cohomology of $\mathbb{\Omega}^{\bullet}(\LQ - d)$.
\end{definition}

\subsection{Total Lagrangian, polarisations and $f$-transformations}
In Remark \ref{Q-dfail} we observed that $L^\bullet$ is generally not an $(\LQ -d)$-cocycle, owing to the \emph{a priori} difference between $L^{\bullet}$ and $L^{\bullet}_{CMR}$, and it is reasonable to ask whether they really differ and how. There are examples in which that is not the case, like in BF theory (cf. Section \ref{BFTheory}), and that happens when $\iota_{Q}\theta^\bullet = L^{\bullet}$. In this section we will see how such difference is in general related to choices of polarisations in the appropriate spaces and, more generally, to structural symmetries of Equations \ref{CMReqts}.

\begin{definition}[BV-BFV Difference] \label{BVBFVDIFFERENCES}
Let $\mathfrak{F}$ be a lax BV-BFV theory. We define the \emph{BV-BFV difference} to be the inhomogeneous local functional:
\begin{equation}
	\bD^{\bullet} := L^{\bullet}_{CMR} - L^{\bullet}\equiv L^{\bullet} - \iota_Q\theta^\bullet.
\end{equation}
Similarly, if $\{M^{(k)}\}$ is a stratification of $M$, we define the \emph{integrated BV-BFV difference at codimension-$k$} or, simply, \emph{$k$-difference} to be
\begin{equation}
	\D^{(k)}\coloneqq\intl_{M^{(k)}} [\bD]^{m-k}.
\end{equation}
\end{definition}

\begin{remark}
Observe now that Eq. \eqref{mCME} can be conveniently written as
\begin{equation}
\LQ L^\bullet = d(L^\bullet + \bD^\bullet).
\end{equation}
\end{remark}

We are ready to give the main result of this section. The data of a lax BV-BFV theory will always produce two $(\LQ -d)$-cocycle, amending the possible breaking of gauge invariance introduced by the presence of a non trivial stratification. Compare this with Section \ref{Sec. WDE}, where $(\LQ -d )$ cocycles are interpreted as solutions of the { descent equations \cite{Z1,MSZ,W1}}.

\begin{theorem}\label{Deltacocycle}
Let $\mathfrak{F}$ be a lax BV-BFV theory. Then, 
\begin{enumerate}
\item $\bD^\bullet$ is an $(\LQ - d)$ cocycle, 
\item $\theta^\bullet$ is a $(\LQ - d)$-cocycle if and only if $\bD^\bullet$ is $\delta$-closed: 
\begin{equation}\label{oneformcocycle}
(\LQ - d)\theta^\bullet = \delta \bD^\bullet,
\end{equation}\label{Eq:Deltacocycle}
\item $L^\bullet\vert_{EL} = \bD^\bullet\vert_{EL}$, 
\item The functional $\mathbb{L}^\bullet\in\mathbb{\Omega}^0(\LQ-d)$
\begin{equation}\label{totalfunctional}
	\mathbb{L}^\bullet\coloneqq L^\bullet + \LE \bD^\bullet \equiv \sum_{k=0}^{m} [L]^{m-k} + k[\bD]^{m-k}
\end{equation}
is a $\left(\LQ-d\right)$-cocycle. 
\end{enumerate}
We will call $\mathbb{L}^\bullet$ the \emph{total Lagrangian associated to the lax BV-BFV theory} $\mathfrak{F}$. 
\end{theorem}
\begin{proof}
The first statement is a consequence of Eq. \eqref{CMReqts}, since ($\LQ=\iota_Q\delta -\delta\iota_Q$)
\begin{multline*}
	\delta L^\bullet = \iota_Q\delta\theta^\bullet - d\theta^\bullet = \LQ\theta^\bullet + \delta\iota_Q\theta^\bullet - d\theta^\bullet \\
	\iff \delta (L^\bullet-\iota_Q\theta^\bullet) = (\LQ - d)\theta^\bullet \iff \delta \bD^\bullet = (\LQ - d)\theta^\bullet
\end{multline*}
then
\begin{equation}
	(\LQ - d)\delta \bD^\bullet = (\LQ - d)^2\theta^\bullet=0 
		\Longrightarrow \delta (\LQ - d)\bD^\bullet = 0
\end{equation}
because $\delta$ commutes with both $d$ and $\LQ$. However, $\bD^{\bullet}$ is a local functional  of (inhomogeneous) degree $|\bD|\geq 0$ and therefore $|(\LQ -d)\bD^\bullet|\geq 1$. This implies that, as there are no nonzero-degree $\delta$-constants, $(\LQ -d)\bD^\bullet=0$. With a similar calculation, using \eqref{mCME} we can gather that 
\begin{equation*}
	\LQ\theta^\bullet = \iota_Q \delta\theta^\bullet - \delta \iota_Q \theta^\bullet 
		= \delta L^\bullet +d\theta^\bullet - \delta\iota_Q\theta^\bullet 
		= \delta\bD^\bullet + d\theta^\bullet
\end{equation*}
whence we conclude that
\begin{equation}
	(\LQ - d) \theta^\bullet = \delta \bD^\bullet.
\end{equation}

Now, observe that the space of solutions $EL$ is defined by the set of equations obtained by setting $Q=0$, so that 
	$$\bD^\bullet\Big\vert_{EL} = (L^\bullet - \iota_Q\theta^\bullet)\Big\vert_{Q=0} = L^\bullet\Big\vert_{EL}.$$

Finally, we compute 
\begin{align*}
	\LQ\mathbb{L}^\bullet & = \LQ L^\bullet + \LQ \LE \bD^\bullet 
		=\LQ L^\bullet + (\LE -1)\LQ\bD^\bullet \\
	& = d(L^\bullet + \bD^\bullet) - d\bD^\bullet + \LE d\bD^\bullet = 
		d ( L^\bullet + \LE \bD^\bullet) = d\mathbb{L}^\bullet
\end{align*}
where we used Lemma \ref{QEuler} and Equation \ref{Eq:Deltacocycle}.
\end{proof}

\begin{remark}\label{polremark}
Following what we observed in the proof of Theorem \ref{Deltacocycle}, that there are no nonzero degree $\delta$-constants, we can also observe that Equation \eqref{oneformcocycle} tells us that in codimension-$k$, with $k\geq1$, $\theta^\bullet$ is an $(\LQ-d)$ cocycle if and only if $\bD^\bullet$ vanishes. Also, observe that computing $\bD^\bullet$ is generally easier than computing $(\LQ -d)\theta^\bullet$.
\end{remark}

\begin{remark}[Polarisations and $f$-transformations]\label{polarisationremark}
It is possible to modify a Lagrangian by a $d$-exact term: $[L]^{k+1} \mapsto [L]^{k+1} + d[f]^{k}$, where $[f]^{k}$ is a $k$-form-valued local functional\footnote{Observe that $[f]^k$ must have ghost degree $(m-k)$, the same of $[\alpha]^{k}$.}, and compensate this with $[\theta]^{k}\mapsto [\theta]^{k} + \delta [f]^{k}$, to preserve equations \eqref{CMReqts}. This is related to the introduction of a polarisation $\mathcal{P}^{(k)}$ on a strictified space of codimension-$k$ fields $\mathcal{F}^{(k)}$, and to the condition that $\alpha^{(k)}$ vanish on the fibres of said polarisation, as required by the quantisation procedure developed\footnote{The quantisation procedure only takes into account boundary polarisations, but generalisations to higher codimensions are being worked out, e.g. in \cite{IM}.} in \cite{CMR2}. In practice, one often starts with a polarisation $\mathcal{P}^{(k)}$ one wants to impose for quantisation, and then one looks for the corresponding $f$-transformation, such that the $f$-transformed $\alpha^{(k)}$ vanishes along (the fibres of) $\mathcal{P}^{(k)}$.
\end{remark}

\begin{definition}[Polarising functionals and $f$-transformations]\label{polarisationdef}
Let $\mathfrak{F}=\left(\mathcal{F},\theta^\bullet, L^{\bullet},Q\right)$ be a lax BV-BFV theory and let $f^\bullet$ be an inhomogeneous local functional of total degree $-1$. We call $f^\bullet$ the \emph{polarising functional} and we define an \emph{$f$-transformation} of the lax BV-BFV theory $\mathfrak{F}$ to be the map 
\begin{equation}
	\mathcal{P}_{f^\bullet}\colon (L^{\bullet},\theta^\bullet)\longmapsto (L^{\bullet} 
		+ df^{\bullet},\theta^\bullet + \delta f^{\bullet}).
\end{equation}
Moreover, if $\mathfrak{F}^{\uparrow n}$ is an $n$-strictification of $\mathfrak{F}$ for the stratification $\{M^{(k)}\}_{k=0\dots m}$, and $f^\bullet$ is a polarising functional, we define the \emph{$f$-transformed difference at codimension-$k$}, or simply \emph{$f$-transformed $k$-difference}, to be 
\begin{equation}
	\D_{f^\bullet}^{(k)} \coloneqq \intl_{M^{(k)}} [\mathcal{P}_{f^\bullet}\bD]^{m-k} \equiv \D^{(k)} 
		+ \intl_{M^{(k)}} [f]^{m-k}.
\end{equation}
\end{definition}

\begin{proposition}\label{PolchangeDeltaL}
An $f$-transformation of a lax BV-BFV theory preserves the class of the BV-BFV difference $\bD^\bullet$ and of the total lagrangian $\mathbb{L}$ in $H^\bullet(\LQ - d)$. In particular:
\begin{equation}\label{poldelta}
	\mathcal{P}_{f^\bullet}(\bD^\bullet) =\bD^\bullet - (\LQ - d) f^\bullet,
\end{equation}
and
\begin{equation}
	\mathcal{P}_{f^\bullet}(\mathbb{L}^\bullet) = \mathbb{L}^\bullet - (\LQ - d)(1+\LE) f^\bullet.
\end{equation}

\end{proposition}

\begin{proof}
From the expression of $\bD^\bullet = L^\bullet - \iota_Q\theta^\bullet$ we immediately compute that
$$
	\mathcal{P}_{f^\bullet}(\bD^\bullet) = \bD^\bullet + df^\bullet - \LQ f^\bullet =
		\bD^\bullet - (\LQ - d) f^\bullet.
$$
Then, from the definition in Eq. \eqref{totalfunctional}, using Lemma \ref{QEuler}, we compute 
\begin{align*}
	\mathcal{P}_{f^\bullet}(\mathbb{L}^\bullet) &= L^\bullet + df^\bullet + \LE(d-\LQ)f^\bullet \\
	&=  L^\bullet + df^\bullet + d\LE f^\bullet - \LE\LQ f^\bullet\\
	&=L^\bullet + d f^\bullet + d\LE f^\bullet - \LQ f^\bullet - \LQ\LE f^\bullet \\
	&= \mathbb{L}^\bullet+ (d-\LQ)f^\bullet + (d-\LQ)\LE f^\bullet \\
	& = \mathbb{L} ^\bullet- (\LQ - d)(1+\LE) f^\bullet
\end{align*}
and clearly $[\mathcal{P}_{f^\bullet}(\mathbb{L}^\bullet)]_{\LQ-d} = [\mathbb{L}^\bullet]_{\LQ-d}$.
\end{proof}

\subsection{BRST construction in the BV setting}\label{Sec:BV/BRST}
Historically, the cohomological approach to field theories has been understood thanks to the work of Becchi, Rouet, Stora and Tyutin \cite{BRS1,BRS2,BRS3,T} who employed the Chevalley--Eilenberg complex to describe invariant functionals on the space of fields. Whenever the (infinitesimal) symmetries of the theory come from a Lie-algebra action, the BRST description is readily available, offering a framework to control gauge fixing and quantisation of the field theory.

Let $(F_M, S^{cl}_M)$ denote the data of a classical theory associated to a space-time manifold $M$, together with symmetries coming from the action of $\mathfrak{g}$ on $F_M$. The BRST complex is understood as the space of functions over the graded manifold $F_{BRST}= F_M \oplus \Omega^0[1](M,\mathfrak{g})$, namely
$$
	C_{BRST}^\bullet = \wedge^\bullet \mathfrak{g}^* \otimes C^\infty(F_M).
$$
We can then interpret the BRST differential as a cohomological vector field $Q_{BRST}$ on $F_{BRST}$.

A BV description of the same data can also easily be obtained.  First, we construct
\begin{equation}
	\mathcal{F}\coloneqq T^*[-1] F_{BRST}
\end{equation}
with its canonical symplectic form $\Omega$. Denoting by $\phi$ the fields in $F_{BRST}$, by $\phi^\dag$ the fields in the cotangent fiber, and by $Q_{BRST}\phi$ the action of the differential on a set of generators $\phi$ of the algebra $C^\infty(F_M)$ (for simplicity, think of the action of $Q_{BRST}$ on fields),  we can construct the functional
$$
	S= S^{cl} + \sum_{\phi}\langle \phi^\dag, Q_{BRST} \phi \rangle
$$
where the angular bracket denotes pointwise pairing of dual fields to produce an $M$-density, and integration over $M$. Then, we can find the \emph{lifted}, Hamiltonian vector field $Q$, i.e. such that
$$
	\iota_{Q}\Omega = \delta S.
$$

\begin{theorem}[\cite{BV81}]
The data $\mathfrak{F}\coloneqq (\mathcal{F},\Omega, S,Q)$ defines a BV theory.
\end{theorem}

\begin{remark}
Observe that in $\mathcal{F}$ there is a preferred Lagrangian submanifold given by the zero-section $F_{BRST}$. It contains classical (degree-$0$) fields and \emph{ghosts}, i.e. the generators of symmetries, in degree-$1$. Consequently, there is a preferred symplectic potential (or Liouville form) 
$$
	\alpha = \sum_{\phi}\langle\phi^\dag, \delta \phi\rangle,
$$
with $\Omega = \delta \alpha$. Then, we observe that
\begin{equation}
	S = S^{cl} + \iota_{Q_{BRST}} \alpha
\end{equation}
Although gauge-fixing in the BV formalism is given by a choice of a Lagrangian submanifold in $\mathcal{F}$, we stress that the choice of the zero section as gauge-fixing is generally not the optimal choice.
\end{remark}

This construction is easily translated in terms of M-form-valued local functionals if we refrain from integrating the above expressions on $M$. In that case we denote the BV Lagrangian density by $L$ and the one-form by $\theta$, with 
$$
	S= \intl_{M} L; \ \ \ \alpha=\intl_{M}\theta.
$$

The reason for this digression on the BRST formalism is the following observation, that will help us understand what is the general meaning of the BV-BFV difference $\bD^\bullet$.

\begin{definition}\label{Def:BRSTTYPE}
Let $\mathfrak{F}$ be a lax BV-BFV theory for the space of fields $\mathcal{F}=T^*[-1]F_{BRST}$, and such that 
\begin{enumerate}
\item $[L]^{\mathrm{top}} = L^{cl} + \iota_Q [\theta]^{\mathrm{top}}$,
\item $[\theta]^{\mathrm{top}} = \phi^\dag\delta\phi$, with $\phi\in F_{BRST}$. Alternatively, we require $[\theta]^{\mathrm{top}}$ to vanish on the fibers of $\mathcal{F} \longrightarrow F_{BRST}$.
\end{enumerate}
We will call such theory: \emph{of BRST type}.
\end{definition}

\begin{theorem}\label{BRSTtypeTheorem}
Let $\mathfrak{F}$ be a lax BV-BFV theory of BRST type. Then we have the following:
\begin{enumerate}
	\item $[\bD]^{\mathrm{top}} = L^{cl}$,
	\item $\mathcal{L}_{Q_{BRST}} L^{cl} = d [\bD]^{\mathrm{top}-1}$,
	\item $\bD^\bullet$ is an $(\mathcal{L}_{Q_{BRST}} - d)$-cocycle; equivalently $\bD^\bullet$ does not depend on the \emph{antifields} $\phi^\dag$ in the cotangent fiber of $T^*[-1]\mathcal{F}_{BRST}$, i.e. $\bD^\bullet$ is the pullback of an element of $\oloc^{0,\bullet}(\mathcal{F}_{BRST}\times M)$.
\end{enumerate}
\end{theorem}

\begin{proof}
To check the first statement we need to compute $[\bD]^{\mathrm{top}} = [L]^{\mathrm{top}} - \iota_Q[\theta]^{\mathrm{top}}$. However, since $\mathfrak{F}$ is of BRST type, we have that 
$$
	[L]^{\mathrm{top}} = L^{cl} + \iota_Q [\theta]^{\mathrm{top}} \Longrightarrow [\bD]^{\mathrm{top}}= [L]^{\mathrm{top}} - \iota_Q[\theta]^{\mathrm{top}} = L^{cl}.
$$
Furthermore, since the theory is lax BV-BFV, in virtue of Theorem \ref{Deltacocycle} we have that $\bD^\bullet$ is an $(\LQ -d)$ cocycle. Hence, we have that 
$$
	\LQ L^{cl} = \LQ [\bD]^{\mathrm{top}} = d[\bD]^{\mathrm{top}-1},
$$
but $\LQ L^{cl} \equiv \mathcal{L}_{Q_{BRST}} L^{cl}$, proving the second statement. Now, since 
$$
	\mathcal{L}_{Q_{BRST}} L^{cl} = d [\bD]^{\mathrm{top}-1}
$$ 
we know that $[\bD]^{\mathrm{top}-1}$ must necessarily be a pullback from $\oloc^{0,\mathrm{top}-1}(\mathcal{F}_{BRST}\times M)$, since $Q_{BRST}$ only involves fields on the base. Then 
$$
	d [\bD]^{\mathrm{top}-2} = \LQ[\bD]^{\mathrm{top}-1} 
		\equiv \mathcal{L}_{Q_{BRST}}[\bD]^{\mathrm{top}-1}
$$
and so on, concluding the proof.
\end{proof}

\begin{remark}
Theorem \ref{BRSTtypeTheorem} is particularly relevant because it characterises $\bD^\bullet$ (in codimension-$1$) as the failure of gauge invariance of the classical action functional under (infinitesimal) gauge transformations. It connects the worlds of BRST and BV, as it allows to construct a $(\mathcal{L}_{Q_{BRST}} -d)$-cocycle from purely BV objects, and then to feed it into an AKSZ machinery, as we will show further on, in the cases of Chern--Simons theory and BF theory, to recover finite gauge transformations.
\end{remark}

\subsection{{ Descent equations}}\label{Sec. WDE} When a field theory is treated in the BRST language, gauge-invariant quantities are phrased in terms of BRST-closed observables, i.e. functionals on fields in the kernel of $\mathcal{L}_{Q_{BRST}}$. A general construction of such observables { was discussed in \cite{Z1,MSZ} and in \cite[Section 3.1]{W1}}, and we will briefly recall the argument here, slightly adapting the language to our setting. Let $\mathcal{F}$ denote sections of a graded vector bundle  $E \to M$ and consider a degree-$1$, cohomological, evolutionary vector field $Q$ on $\mathcal{F}$. In the original setting $Q$ denotes the BRST operator, but we will work with its BV extension (cf. Section \ref{Sec:BV/BRST}). Our goal is to construct $Q$-closed  observables $\O \in \Omega^0_{\mathrm{loc}}(\mathcal{F})$.  Suppose we can find some $\O^{(0)} \in \Omega^{0,0}_{\mathrm{loc}}(\mathcal{F}\times M)$ such that\footnote{In what follows, the superscript $^{(k)}$ will denote form degree instead of the previously used co-degree.} 
\begin{equation}
	\LQ \O^{(0)} = 0. 
\end{equation}
Specifying a point $x \in M$ we obtain a $Q$-closed observable $\O^{(0)}(x)\colon \mathcal{F} \to \mathbb{C}$. It is a natural question whether the BRST/BV cohomology class of this observable depends on the point $x$. Let $x'$ be another point in $M$ and $\gamma$ be a 1-chain with boundary $\partial \gamma = x - x'$. Then 

\begin{equation}
	\O^{(0)}(x) - \O^{(0)}(x') = \int_\gamma d\O^{(0)} 
\end{equation} 
is trivial in the $Q$-cohomology if and only if there exists $\O^{(1)} \in \Omega_{\mathrm{loc}}^{0,1}(\mathcal{F}\times M)$ such that 
\begin{equation}
	\LQ \O^{(1)} = d\O^{(0)}. \label{eq:WittenDescent1}
\end{equation}
We call Equation \eqref{eq:WittenDescent1} a { \emph{descent equation}}. Notice that in this case we can produce a new observable $\O^{(1)}(\beta)$ by integrating over any 1-cycle $\beta$: 
\begin{equation}
	\LQ \int_{\beta}\O^{(1)} = \int_{\beta}d \O^{(0)} = 0.
\end{equation}
Again, one can ask whether the $Q$-cohomology class of this observable is independent of the representative of the homology class $[\beta]$. 
This is the case if and only if there is $\O^{(2)} \in \Omega_\mathrm{loc}^{0,2}(\mathcal{F}\times M)$ such that 
\begin{equation}
	\LQ \O^{(2)} = d\O^{(1)}.
\end{equation}
Proceeding in this manner, one can produce observables $\O^{(k)}(\gamma^{(k)})$ whose $Q$-cohomology class depends only on the homology classes of $\gamma$, provided that we can solve at every $k=0,\ldots,\dim M$ the { descent equation} 
\begin{equation}
	\LQ \O^{(k+1)} = d\O^{(k)}.
\end{equation} In \cite{W1} it is argued that the expectation values of such observables (in the quantum theory) produce topological invariants. 
This motivates the following definition.
\begin{definition} 
Let $Q$ be a cohomological, evolutionary vector field on $\mathcal{F}$. Then we say that a functional $\O^{\bullet} \in \Omega^{0,\bullet}(\mathcal{F}\times M)$ \emph{satisfies the { descent equations}} if 
\begin{equation}
	(\LQ - d)\O^{\bullet} = 0. 
\end{equation}
\end{definition}
The { descent equations} are nothing but the cocycle condition in the complex $\mathbb{\Omega}^\bullet(\LQ-d)$. Thus, we observe that $(\LQ-d)$-cocycles can be naturally paired to cycles to produce $\LQ$-closed observables. 
\begin{remark}
The { descent equation} $\LQ \O^{(k+1)} = d\O^{(k)}$ can, of course, be read both ways. It is equally natural to start with a functional $\O^{(\mathrm{top})} \in \Omega^{0,\mathrm{top}}_\mathrm{loc}(\mathcal{F}\times M)$ and try to extend ``downwards''. This is the point of view taken in \cite{Alekseev} and also in the present paper.
\end{remark}
\begin{remark}
Given any lax BV-BFV theory, there are two natural $(Q-d)$-cocycles associated to it: The difference $\bD^\bullet$ and the total Lagrangian $\mathbb{L}^\bullet$. Thus extended BV-BFV theories naturally produce solutions to the { descent equations}. In Theorem \ref{BRSTtypeTheorem} we have shown how, if the theory is BRST-type, we can always construct a solution of { descent} for the BRST operator.
\end{remark}

\subsection{AKSZ Formalism}\label{Sec:AKSZ}
The AKSZ formalism (after Alexandrov, Kontsevich, Schwarz and Zaboronski \cite{AKSZ}) is a general construction to produce BV theories, which is - in particular - compatible with the BV-BFV axioms \cite[Section 6]{CMR} in the case of field theories on manifolds with boundaries and corners\footnote{This means that the outcome of the AKSZ procedure is, usually, a strict fully-extended BV-BFV theory.}. To an $n$-dimensional, ordinary manifold $N$ and an $(n-1)$-symplectic, graded manifold $(X,\omega)$ endowed with a degree$-n$ function $\vartheta$ satisfying the classical master equation\footnote{We denote by $ \{\cdot,\cdot \}_\omega$ the Poisson bracket induced by $\omega$.} $\{\theta,\theta\}_\omega=0$ and (possibly) a degree$-(n-1)$ one-form $\alpha$, it associates the AKSZ space of fields $\mathcal{F}^{AKSZ}\coloneqq \mathrm{Map}(T[1]N,X)$ and defines the following.
\begin{definition}[Transgression map]\label{transgressionmap}
Consider the map
\begin{equation}\label{transgression}
	\T^\bullet_N \colon \Omega^\bullet(X) \longrightarrow \Omega^\bullet\left( \mathrm{Map}(T[1]N,X)\right)
\end{equation}
defined by $\T^\bullet_N\coloneqq p_*\mathrm{ev}^*$, where
\begin{equation}
	\xymatrix{
		\mathrm{Map}(T[1]N, X)\times T[1]N \ar[d]_{p}   \ar[r]_-{\mathrm{ev}}  & X\\
		\mathrm{Map}(T[1]N, X) & 
	}.
\end{equation}
We will call $\T^\bullet_N$ the \emph{transgression map}, and its evaluation a \emph{transgression}. For notational purposes, we will denote by $\T\equiv\T^0_N$ the transgression map on functionals.
\end{definition}

Then we have
\begin{theorem}[\cite{AKSZ}]\label{AKSZ}
The data 
\begin{equation}\label{AKSZdata}
	\mathfrak{F}^{AKSZ}\coloneqq\left(\mathcal{F}^{AKSZ}, \Omega^{AKSZ}, S^{AKSZ}, Q^{AKSZ}\right)
\end{equation}
defines a BV theory, with $\Omega^{AKSZ}\coloneqq \T^{(2)}_N(\omega)$, the deRham differential $d_{N}$ on $N$ seen as a degree$-1$ vector field on $\mathcal{F}^{AKSZ}$, a functional $S^{AKSZ}\colon \mathcal{F}^{AKSZ}\to \mathbb{R}$,
\begin{equation}
	S^{AKSZ}\coloneqq\T_N^{(0)}(\vartheta) + \iota_{d_{N}}\T_N^{(1)}(\alpha).
\end{equation}
and a cohomological vector field $Q^{AKSZ}$ such that $\iota_{Q^{AKSZ}}\Omega^{AKSZ}=\delta S^{AKSZ}$.
\end{theorem}

\begin{remark}
Observe that Theorem \ref{AKSZ} involves real-valued local functionals, as opposed to densities. This is the standard setting for BV theory, and it usually assumes integrating relevant densities, hence strictifying the BV-BFV data.
\end{remark}

Consider now a (1-extended) BV-BFV theory, where the (boundary) BFV data is given by $(\mathcal{F}^{(1)},\alpha^{(1)}, L^{(1)},Q^{(1)})$. We can apply the AKSZ construction considering the graded (super-)manifold $\mathcal{F}^{(1)}$ as our target, endowed with a degree-1 local functional, with source manifold the interval $I=[a,b]$. In other words we consider
\begin{equation}
	\mathcal{F}^{AKSZ}=\mathrm{Map}(T[1]I,\mathcal{F}^{(1)})
\end{equation}
The natural choice of a functional on $\mathcal{F}^{(1)}$ is indeed $L^{(1)}$ and, together with the given BFV one-form $\alpha^{(1)}$, we can produce a BV theory following Theorem \ref{AKSZ}. This, in particular, defines the \emph{AKSZ critical locus}, i.e. the set of critical points of $S^{AKSZ}$, which is given by differential graded maps: 
\begin{equation}
	\mathrm{EL}_{AKSZ}\coloneqq\mathrm{Crit}(S^{AKSZ})=\mathrm{dgMap}(T[1]I,\mathcal{F}^{(1)}).
\end{equation}

However, we could define a somewhat larger critical locus by retaining only the 1-form component (along I) of the EL equation\footnote{This is equivalent to considering variational derivatives of the AKSZ action functional only with respect to fields in $\mathrm{Map}(I,\mathcal{F}^{(1)})$.}. 

\begin{definition}\label{ELlocus}
We define the \emph{transversal Euler--Lagrange locus} associated to the AKSZ theory $(\mathcal{F}^{AKSZ},S^{AKSZ})$ to be the space of solution of the field equations coming from the variations of $S^{AKSZ}$ with respect to fields in $\mathrm{Map}(I,\mathcal{F}^{(1)})$.

We will denote this enlarged locus by $\mathrm{dgMap}_I(T[1]I,\mathcal{F}^{(1)})$ and its restriction to degree-zero maps by $\mathrm{dgMap}_I^0(T[1]I,\mathcal{F}^{(1)})$.
\end{definition}

\begin{remark}
Let $\mathfrak{F}^{\uparrow n}=\left(\mathcal{F}^{(k)}, \alpha^{(k)}, S^{(k)}, Q^{(k)}\right)$ be an $n$-extended theory coming from the AKSZ construction. Observe that $S^{AKSZ}= \T_I^{(0)}S^{(1)} + \iota_{d_{I}}\T_I^{(1)}\alpha^{(1)}$, on $\mathrm{Map}(T[1]I,\mathcal{F}^{(1)}_{CS})$ - for the target manifold $\mathcal{F}^{(1)}_{CS}$, the codimension-$1$ strictified space of states for Chern Simons theory, with $S^{(1)}_{CS}$ and $\alpha_{CS}^{(1)}$ the associated strict, boundary action and one-form - reproduces Chern--Simons theory on the cylinder $M^{(1)}\times I$. A similar construction was presented in \cite{CStime}, for one-dimensional reparametrisation-invariant models.
\end{remark}

\section{Chern--Simons Theory}
Chern--Simons theory can be seen as a fully extended BV-BFV theory on a three-dimensional manifold $M$. The space of fields on $M$ is given by Lie algebra-valued inhomogeneous forms $\cA\in\Omega^\bullet(M)[1 - \bullet]\otimes\mathfrak{g}$, and the degree-zero part of the theory is the usual Chern--Simons theory of connections on a (trivial) principal bundle $P\longrightarrow M$. 

In this section we will analyse the information one can extract from its BV-BFV description in higher codimensions. We explicitly connect the BFV boundary data with the well-known Wess--Zumino and Wess--Zumino--Witten functionals by means of a new construction that adjoins AKSZ collars to boundaries, effectively performing an integration of Lie algebra valued fields to Lie group valued ones.

The example of Chern--Simons will serve as a guideline to generalise to other field theories (cf. Section \ref{BFTheory}) and will set the expectations on what we can predict by analysing examples that are not as well studied as this one.

We will discuss four polarising functionals and the respective $f$-transformations of lax CS theory (see Definition \ref{polarisationdef}). The first two, denoted by $f_{min}$ and $f_{min}^{1,0}$, will represent the minimal modifications one needs in order to reproduce gauged Wess--Zumino(--Witten) functionals (respectively, see Definition \ref{WZWdef}), through the AKSZ integration procedure (see Section \ref{Sec:AKSZ} and Theorem \ref{Theorem:WZWfromDelta}). The third and fourth ones, denoted $f_{tot}$ and $f_{tot}^{1,0}$ will also put Chern--Simons theory in its BRST form (Definition \ref{Def:BRSTTYPE}, Proposition \ref{CSBRSTtype}).

\subsection{Generalities}
We fix the notation following \cite{CMR}:

\begin{propdef}\label{ChernSimonsTheory}
Consider a connected, simply connected Lie group $G$, with $\left((\mathfrak{g},\ [\cdot,\cdot]),\ \langle \cdot,\cdot \rangle\right)$ its Lie algebra endowed with an invariant nondegenerate inner product. Then, the data
	$$\mathcal{F}_{CS}\coloneqq\Omega^\bullet(M)[1-\bullet]\otimes\mathfrak{g}\ni\cA$$
together with $L_{CS}^\bullet\in\oloc^{0,\bullet}(\mathcal{F}_{CS})$ and $\theta_{CS}^{\bullet}\in\oloc^{1,\bullet}(\mathcal{F}_{CS})$ given by, respectively
\begin{subequations}\label{CSdata}\begin{align}\label{CSlagrangian}
	L_{CS}^{\bullet}[\cA]  & \coloneqq \frac12 \langle\cA, d\cA\rangle + \frac16 \langle\cA, [\cA,\cA]\rangle\\\label{CSoneform}
	\theta^\bullet[\cA] & \coloneqq  \frac12 \langle\cA, \delta\cA\rangle,
\end{align}\end{subequations}
and a vector field\footnote{In fact one needs to take the infinite prolongation of $Q_{CS}$; this step is always implied.} $Q_{CS}$ such that
\begin{equation}
	\mathcal{L}_{Q_{CS}}\cA \coloneqq d\cA + \frac12 [\cA,\cA] ,
\end{equation}
defines a lax BV-BFV  theory. We will call the data
\begin{equation}\mathfrak{F}_{CS}=\left(\mathcal{F}_{CS}, L^{\bullet}_{CS}, \alpha_{CS}^{\bullet}, Q_{CS}\right)\end{equation}
\emph{lax Chern--Simons theory}.
\end{propdef}

We shall omit the pairing symbol $\langle\cdot,\cdot\rangle$ from now on. 

\begin{proof}
We need to check that, with the above definitions, Equations \eqref{CMReqts} are satisfied. Let us compute:
$$
\iota_Q\delta\theta^\bullet = \frac12 \iota_Q(\delta\cA\delta\cA)= d\cA \delta\cA + \frac12[\cA,\cA]\delta\cA
$$
on the other hand, recalling that $\cA$ has total degree $1$,
\begin{align*}
\delta L^\bullet + d\theta^\bullet & = \delta\left(\frac12 \cA d\cA + \frac16 \cA[\cA,\cA]\right) + \frac12 d\left(\cA\delta\cA\right) \\
& = \frac12\delta\cA d\cA + \frac12 \cA d\delta\cA + \frac12 \delta\cA[\cA,\cA] + \frac12d\left(\cA\delta\cA\right) = d\cA \delta\cA + \frac12[\cA,\cA]\delta\cA
\end{align*}
where we expanded $d\left(\cA\delta\cA\right) = d\cA\delta\cA - \cA d \delta\cA$, showing that \eqref{BVBFVrel} holds. To check Equation \eqref{bracket}, we compute
\begin{align*}
\iota_Q\iota_Q\delta\theta^\bullet & = \iota_Q\left(d\cA \delta\cA + \frac12[\cA,\cA]\delta\cA\right) \\
& = d\cA d\cA + \frac12 d\cA[\cA,\cA] + \frac12[\cA,\cA]d\cA + \frac14[\cA,\cA][\cA,\cA]\\
&= d\left(\cA d\cA + \frac13\cA[\cA,\cA]\right) = 2dL^\bullet,
\end{align*}
where we used Jacobi identity in $[\cA,[\cA,\cA]]=0$.
\end{proof}

\begin{lemma}
The modified Lagrangian for lax Chern--Simons theory is given by
\begin{equation}
	L_{CMR}^{\bullet}[\cA]= 2L^{\bullet}_{CS} - \iota_{Q}\theta^\bullet = \frac12 \cA d\cA + 
		\frac{1}{12}\cA[\cA,\cA],
\end{equation}
the BV-BFV difference reads
\begin{equation}\label{Eq:CSDeltacocycle}
	\bD_{CS}^{\bullet} = - \frac{1}{12}\cA[\cA,\cA],
\end{equation}
and the total lagrangian reads
\begin{align}
	\mathbb{L}^\bullet_{CS} & = L^{\bullet}_{CS} + \LE \bD_{CS}^{\bullet}\label{eq:LCS}\\\notag
		& = \frac12 \left( AdA + A^\dag dc + cdA^\dag + cdA + Adc + cdc \right)\\\notag
		& + \frac16 A[A,A] + A^\dag[A,c] + \frac12 c^\dag[c,c] + \frac14 c[A,A] + \frac14 A^\dag[c,c] + \frac12 c[c,c]
\end{align}
with the decomposition
$$\cA=c + A + A^\dag + c^\dag\in \Omega^0(M)[1]\times \Omega^1(M)[0] \times \Omega^2(M)[-1]\times \Omega^3(M)[-2].$$
\end{lemma}

\begin{proof}
This is just a matter of straightforward computations.
\end{proof}

Lax BV-BFV Chern--Simons theory can be made strict and fully extended, as was directly shown in \cite{CMR}. The strictification singles out the homogeneous parts of $L^\bullet_{CS}$ and $\theta^\bullet_{CS}$ and integrates over the appropriate stratum.

\begin{theorem}[\cite{CMR}]\label{StrictCSTheorem}
Lax Chern--Simons theory defines a fully-extended BV-BFV theory on every stratification $\{M^{(k)}\}$ of $M$, by the following data:
\begin{enumerate}
	\item $\mathcal{F}^{(k)}=\Omega^\bullet(M^{(k)})[1-\bullet]$ with $\pi_{(k)}=\iota_{(k)}^*\colon \mathcal{F}^{(k)}\longrightarrow \mathcal{F}^{(k+1)}$,
	\item $Q^{(k)} = (\pi_{(k)})_*Q_{CS}$,
	\item $\alpha^{(k)} = \intl_{M^{(k)}}[\theta_{CS}]^{m-k}$ and $S^{(k)}=\intl_{M^{(k)}}[L_{CS}]^{m-k}$,
	\end{enumerate}
together with $p_{(k)}\colon \mathcal{F}_{CS}\equiv\mathcal{F}^{(0)}\longrightarrow \mathcal{F}^{(k)}$ the composition of all $\pi_{(\leq k)}$.
\end{theorem}

\subsection{Polarising functionals for CS theory}\label{boundarypolarisation}
Following Theorem \ref{StrictCSTheorem}, the space of codimension-$1$ (boundary) fields $\mathcal{F}_{CS}^{(1)}$ is given by pullback of fields to the stratum along $\iota_{(1)}\colon M^{(1)} \rightarrow M^{(0)}$, and the pair $(\mathcal{F}_{CS}^{(1)},\omega^{(1)}=\delta \alpha^{(1)})$ is an exact $0$-symplectic manifold\footnote{A more general situation is when $\alpha^{(1)}$ is not a one-form but a connection on a line bundle. Then $\omega^{(1)}$ is interpreted as its curvature.}. For the sake of quantisation, one might be interested in choosing a polarisation on $\mathcal{F}^{(1)}$ and would be required to modify the boundary one form $\alpha^{(1)}$ so that it vanishes on the (Lagrangian) fibres of said polarisation. 

In order to do this, we pick a complex structure on the (two-dimensional) stratum $M^{(1)}$, which induces a splitting of the the space of 1-forms into its Dolbeault parts 
$$\Omega^1(M^{(1)}) = \Omega^{1,0}(M^{(1)}) \oplus \Omega^{0,1}(M^{(1)}),$$
where $\Omega^{1,0}(M^{(1)})$ is the space of 1-forms that locally look like  $\beta(z,\bar{z})dz$. Then, the space of boundary fields splits as 
\begin{equation}\label{Dolbeaultsplitting}
	\mathcal{F}_{CS}^{(1)} = \Omega^0(M^{(1)},\g) \oplus \Omega^{1,0}(M^{(1)},\g) 
		\oplus \Omega^{0,1}(M^{(1)},\g) \oplus \Omega^{2}(M^{(1)},\g),
\end{equation}
and the connection field $A$ on the stratum\footnote{We use the same symbol for $A$ and $\iota_{(1)}^*A$, as there should be no source of confusion.} $M^{(1)}$ splits into its holomorphic and anti-holomorphic parts (resp. $A^{1,0}$ and $A^{0,1}$). The fibres of the polarisation will be defined by constant $A^{1,0}$ and $c$, thus defining the Lagrangian fibration
\begin{equation*}
	\mathcal{F}^{(1)}_{CS} \longrightarrow \Omega^0(M^{(1)},\g) \oplus \Omega^{1,0}(M^{(1)},\g)
\end{equation*} 

We will need the following definitions.
\begin{definition}\label{Polarisingfunctionals}
We define the polarising functionals $f_{min}, f_{tot}, f_{min}^{1,0}, f_{tot}^{1,0}$ to be 
\begin{subequations}\begin{align}
	f_{min}& \coloneqq \frac12 cA^\dag\\
	f_{tot}& \coloneqq  \frac12 \left(AA^\dag + cc^\dag + cA^\dag + cA\right)\\
	f_{min}^{1,0}&\coloneqq \frac12 \left( A^{1,0}A^{0,1} + cA^\dag \right)\\
	f^{1,0}_{tot} &\coloneqq \frac12 \left(AA^\dag + cc^\dag + A^{1,0}A^{0,1} + cA^\dag + cA\right),
\end{align}\end{subequations}
where the superscript ${}^{1,0}$ refers to the splitting in \eqref{Dolbeaultsplitting} and depends on the data of a complex structure on $M^{(1)}$.
\end{definition}

\begin{remark}
The reason underlying the nomenclature we used in Definition \ref{Polarisingfunctionals} will become clearer as we proceed. The choice of utilising a complex structure in $M^{(1)}$ to define the polarising functional is, once again, related to the choice of a polarisation in $F^{(1)}_{CS}$ that depends on such complex structure. 
\end{remark}

\begin{definition}\label{polarisedCSTheory}
We define \emph{$f$-transformed, lax Chern--Simons theory} to be the BV-BFV data obtained from the Chern--Simons BV-BFV data by the $f$-transformation
\begin{equation}
	\mathcal{P}_{f^\bullet}(L^{\bullet}(k),\theta^\bullet(k)) 
\end{equation}
with polarising functional $f^\bullet\in\{f_{min}, f_{tot}, f_{min}^{1,0}, f_{tot}^{1,0}\}$. The $f$-transformation changes the representative of $[\mathbb{L}^\bullet]_{\LQ -d}$ to the $f$-transformed, total Lagrangian
\begin{equation}
\mathcal{P}_{f^\bullet}(\mathbb{L}^\bullet) = \mathbb{L}^\bullet - (\LQ - d)(\LE+1)f^\bullet.
\end{equation}
\end{definition}

We have:

\begin{proposition}\label{CSBRSTtype}
The \emph{$f$-transformed lax BV-BFV Chern--Simons theory}, given by $\mathcal{P}_{f^\bullet}(L^\bullet_{CS},\theta^\bullet_{CS})$ is of BRST type for $f^\bullet\in\{f_{tot},f^{1,0}_{tot}\}$, as in Definition \ref{Polarisingfunctionals}. Moreover the \emph{$f^\bullet$-transformed BV-BFV differences} read
\begin{subequations}\begin{align}\label{BRSTtypeDeltacocycle}
	\mathcal{P}_{f_{tot}}\bD^\bullet_{CS} 
		&= L_{CS}^{cl} + \frac12Adc + \frac12cdc -\frac{1}{12} c[c,c],\\
	\mathcal{P}_{f^{1,0}_{tot}}\bD^\bullet_{CS} 
		&= L_{CS}^{cl;1,0} + A^{1,0}\delbar c + \frac12cdc -\frac{1}{12} c[c,c],
\end{align}\end{subequations}
where the classical Chern--Simons Lagrangians are given by 
$$
	L_{CS}^{cl} = \frac12AdA + \frac16 A[A,A],
$$
and
$$
	L_{CS}^{cl;1,0} = \frac12AdA + \frac16 A[A,A] + \frac12d(A^{1,0}A^{0,1}).
$$
\end{proposition}

\begin{proof}
From Equation \eqref{CSlagrangian} we extract
$$
	[L_{CS}]^{\mathrm{top}} = \frac12 \left(AdA + A^\dag dc + cdA^\dag\right) 
		+ \frac16 A[A,A] + \frac12c^\dag[c,c] +A^\dag[A,c]
$$
so that, from $[\mathcal{P}_{f_{tot}}L_{CS}]^{\mathrm{top}} = L_{CS}^{cl} + \iota_Q[\mathcal{P}_{f_{tot}}\theta]^{\mathrm{top}}$, we have
$$
	[\mathcal{P}_{f_{tot}}L_{CS}]^{\mathrm{top}} = [L_{CS}]^{\mathrm{top}} + \frac12d[cA^\dag] 
		= \frac12AdA + \frac16 A[A,A] + A^\dag d_A c + \frac12c^\dag[c,c].
$$
Moreover, from Equation \eqref{CSoneform} we have
$$
	[\mathcal{P}_{f_{tot}}\theta]^{\mathrm{top}} = [\theta]^{\mathrm{top}} + \frac12\delta(AA^\dag + cc^\dag) = A^\dag \delta A + c^\dag \delta c.
$$
It is a matter of a simple computation to check the explicit formula for $\mathcal{P}_{f_{tot}}\bD^\bullet_{CS}$. An analogous calculation can be performed for the $1,0$-case.
\end{proof}

\begin{remark}
Compare Equation \eqref{BRSTtypeDeltacocycle} with the third statement of Theorem \ref{BRSTtypeTheorem}. As expected, the BV-BFV difference was translated into an $(\mathcal{L}_{Q_{BRST}} - d)$-cocycle.
\end{remark}

We conclude this section with the following observation.

\begin{lemma}
The $f$-transformed differences at codimension-$1$, with the polarising functionals given in Definition \ref{Polarisingfunctionals} are given by
\begin{subequations}\begin{align}\label{Deltasemipolarised}
	\D_{f_{min}}^{(1)} & = \frac12 \intl_{M^{(1)}} Adc  - d(cA),\\
	\D_{f_{tot}}^{(1)} & = \frac12 \intl_{M^{(1)}} Adc,\\
	\D^{(1)}_{f_{min}^{1,0}} & = \intl_{M^{(1)}} A^{1,0}\delbar c - \frac12 d(cA),\\
	\D^{(1)}_{f_{tot}^{1,0}} & = \intl_{M^{(1)}} A^{1,0}\delbar c
\end{align}\end{subequations}
\end{lemma}
\begin{proof}
This is a straightforward computation from 
$$
	\D^{(1)}_{f^\bullet}=\D^{(1)} - \intl_{M^{(1)}}(\LQ-d) f^\bullet 
		= \D^{(1)} - \intl_{M^{(1)}}\LQ [f]^2 - d [f]^1,
$$
recalling that (cf. Equation \eqref{Eq:CSDeltacocycle})
$$
	\D^{(1)}=\intl_{M^{(1)}} [\bD]^2 = -\frac14 \intl_{M^{(1)}}  [c,c]A^\dag + c[A,A],
$$
and observing that $\LQ(A^{1,0}A^{0,1}) = \del_{A^{1,0}}cA^{0,1} - A^{1,0}\delbar_{A^{0,1}}c = \del cA^{0,1} - A^{1,0}\delbar c$.
\end{proof}

\begin{remark}
Notice that if $M^{(1)}$ has empty boundary, the formulas for $\D^{(1)}_{f_{min}}$ and $\D^{(1)}_{f_{tot}}$ are indistinguishable, and similarly for their $(1,0)$-analogues.
\end{remark}

\subsection{Wess--Zumino--Witten theory from Chern--Simons theory}\label{Sec:CS-WZW}
A $k$-form valued local functional like the codimension-$k$ Lagrangian $[L]^{k}$ can be integrated on the $k$-stratum to yield a local functional. In this section we will be mostly concerned with the \emph{Chern--Simons action functionals} coming from Definitions \ref{ChernSimonsTheory} and \ref{polarisedCSTheory}.
\begin{equation}\label{CSaction}
	S\coloneqq\intl_{M} [L_{CS}]^{\mathrm{top}};\ \ \ \ S^{1,0}\coloneqq \intl_{M} [\mathcal{P}_{f_{min}^{1,0}}(L_{CS})]^{\mathrm{top}} 
		= S + \intl_{M} df_{min}^{1,0}.
\end{equation}

The structure group $G$ acts on the space of fields by means of \emph{finite gauge transformations}, as follows.

\begin{definition}\label{largegauge}
A \emph{finite gauge transformation} is the (right) action of a group valued field $g\in C^\infty(M,G)$ on the space of fields. Connections $A$ are acted upon via 
	$$(A)^g = g^{-1}Ag + g^{-1}dg.$$
Introducing the splitting discussed in Section \ref{boundarypolarisation}, the $1,0$ and $0,1$ parts of $A$ on $\partial M$ transform as 
\begin{align*}
	(A^{1,0})^g &= g^{-1}A^{1,0}g + g^{-1}\de g \\
	(A^{0,1})^g &= g^{-1}A^{0,1}g + g^{-1}\bar\de g 
\end{align*}
whereas the remaining fields in $\mathcal{F}$ transform as 
\begin{align*}
	(A^\dag)^g = & g^{-1}A^\dag g\\
	(c)^g = & g^{-1}c g\\
	(c^\dag)^g = & g^{-1}c^\dag g.
\end{align*}
Finally, we declare the action of $h\in C^\infty(M,G)$ on $g$ as 
$$
	h\cdot g = h^{-1} g.
$$
\end{definition}

\begin{definition}\label{WZWdef}
Let $(M,\de M)$ be a three-dimensional manifold with boundary, $A\in\mathfrak{A}_{\partial M}$ be a connection on a (trivial) principal bundle $P^\partial\longrightarrow \partial M$ and $\tilde{g} \in C^{\infty}(M,G)$ an arbitrary extension of $g\in C^\infty(\partial M, G)$, i.e. $\tilde{g}\vert_{\partial M} = g$. We define the space of Wess--Zumino fields to be $\mathcal{F}_{WZ}(\partial M)=\mathfrak{A}_{\partial M} \times C^\infty (\partial M,G)$ and, on it, the \emph{Wess--Zumino functional}\footnote{Notice that Equation \eqref{WZ} is well-defined modulo $4\pi^2 \mathbb{Z}$ (for the standard normalization of the Killing form on $\mathfrak{g}$ and of the Cartan 3-form on $G$), in the case of a \emph{compact, simple} Lie group $G$. Another example we will need for Section \ref{BFTheory} is a group of the form $\widetilde{G}=G\times \mathfrak{g}^*$. $\widetilde{G}$ is neither simple nor compact, but the WZ term itself is well-defined (even without mod $\mathbb{Z}$) and in fact can be written as a surface (rather than bulk) integral, since the Cartan 3-form in this case is exact. Observe, furthermore, that the standard normalization of the gauged WZW action functionals in the literature times $2\pi i$ (see e.g. \cite{GK}) recovers the one presented here, where a different convention on group actions is used.} (WZ):
\begin{equation}\label{WZ}
	S_{WZ}[g]=\frac{1}{12} \int_M \langle \tilde{g}^{-1}d\tilde{g},[\tilde{g}^{-1}d\tilde{g},\tilde{g}^{-1}d\tilde{g}]\rangle,
\end{equation}
the \emph{gauged Wess--Zumino} (gWZ) functional:
\begin{equation}
	S_{gWZ}[A,g] \coloneqq \frac{1}{2}\int_{\de M} \langle A,dg\, g^{-1} \rangle 
		-S_{WZ}[g],
\end{equation}
and, given a complex structure on $\partial M$, the \emph{gauged Wess--Zumino--Witten} (gWZW) functional:
\begin{equation}
	S_{gWZW}^{1,0}[A^{1,0},g]= \int_{\de M} \langle A^{1,0},\delbar g\, g^{-1} \rangle 
		+ \frac{1}{2} \langle g^{-1} \de g, g^{-1}\bar\de g\rangle - S_{WZ}[g].
\end{equation}
\end{definition}

The proofs of the following statements are given in Appendix \ref{A:ProofsCS-WZW}.

\begin{lemma}\label{Lemma:WZWfailure}
Let $S$ denote the BV Chern--Simons action functional of Eq. \eqref{CSaction} and let $g \in C^{\infty}(\partial M,G)$ generate the finite gauge transformation of Definition \ref{largegauge}. Then, we have (cf. Equation \eqref{CSaction}) 
	$$S[\calA^g] - S[\calA]  =  S_{gWZ},$$ 
together with
	$$S^{1,0}[\calA^g]-S^{1,0}[\calA] = S_{gWZW}^{1,0}. $$
\end{lemma}

\begin{lemma}\label{Lemma:gWZWinvariance}
Consider the gauged Wess--Zumino and gauged Wess--Zumino--Witten functionals $S_{gWZW}[A,g], S^{1,0}_{gWZW}[A^{1,0},g]$, and a finite gauge transformation generated by $h\in C^\infty(\partial M,G)$. Then
\begin{subequations}\begin{align}\label{S_gWZ gt property}
	S_{gWZ}[A^{h} , h\cdot g] & = S_{gWZ}[A,g] - S_{gWZ}[A,h]\\
	S^{1,0}_{gWZW}[(A^{1,0})^{h} , h\cdot g] & = S^{1,0}_{gWZ}[A^{1,0},g] 
		- S^{1,0}_{gWZ}[A^{1,0},h].
\end{align}\end{subequations}
In particular, this implies that the functionals
\begin{subequations}\begin{align}
	S_{\mathrm{inv}}&\coloneqq S_{CS}[\cA] + S_{gWZ}[A,g];  \\
	S^{1,0}_{\mathrm{inv}} & \coloneqq S^{1,0}_{CS}[\cA] + S^{1,0}_{gWZW}[A^{1,0},g]
\end{align}\end{subequations}
are invariant under finite gauge transformations.
\end{lemma}

\begin{lemma}\label{Lemma:ELtransverse}
Let $\gamma_t$ be a time-dependent path in $C^\infty(\partial M,\mathfrak{g})$ and define $g_t\coloneqq\mathrm{Pexp}(\int_0^t \gamma_s ds)$, with initial condition $g_0$. Then:
\begin{equation}
	\frac{d}{dt} A^{g_t} = d_{A^{g_t}}\gamma_t.
\end{equation}
where $A^{g_t}=g^{-1}_tA\,g_t+ g^{-1}_tdg_t$. Similarly for $A^{1,0}$, replacing $d\gamma_t$ with $\de\gamma_t$. Moreover, if $\phi_t=g_t^{-1} dg_t$,
\begin{equation}\label{maurercartanvariation}
	\frac{d}{dt} \phi_t = d_{\phi_t}\gamma_t.
\end{equation}
\end{lemma}

\begin{lemma}\label{Lemma:threedterm}
Let $g_t$ as in Lemma \ref{Lemma:ELtransverse} then, for every extension $\widetilde{g}_t\in C^\infty(M,G)$ with $\widetilde{g}_t\vert_{\partial M} = g_t$, we have 
\begin{equation}\label{phithreetermvariation}
	-\frac{d}{dt} \int_M \frac{1}{12}\langle \widetilde{g}_t^{-1}d\widetilde{g}_t,[\widetilde{g}_t^{-1}d\widetilde{g}_t,\widetilde{g}_t^{-1}d\widetilde{g}_t]\rangle 
		= \frac12 \intl_{\partial M} \langle g_t^{-1} dg_t, d\gamma_t\rangle.
\end{equation}
\end{lemma}

\begin{proposition}\label{Proposition:WZderivatives}
Let $g_t$ be a time-dependent family of group valued functions $C^\infty(M,G)$ of the form $g_t=\mathrm{Pexp}(\int_0^t \gamma_s ds)$, where $\gamma_t\in\Omega^0(M,\mathfrak{g})$ for all $t$. Then
\begin{subequations}\begin{align}
	\frac{d}{dt} S_{gWZ}[A,g_t] =& \frac12 \int_{\de M} \langle A^{g_t}, d \gamma_t\rangle \\
	\frac{d}{dt} S^{1,0}_{gWZW}[A^{1,0},g_t] =& \int_{\de M} \langle (A^{1,0})^{g_t},\bar\de\gamma_t\rangle.
\end{align}\end{subequations}
\end{proposition}

We are now ready to state the main result of this section.

\begin{theorem}\label{Theorem:WZWfromDelta}
Consider Chern--Simons theory on a manifold with boundary $(M,\partial M)$ for the connected, simply connected structure group $G$. Let $\mathrm{Map}(T[1]I,\mathcal{F}^{(1)}_{CS})$ be the AKSZ space of fields and $\T$ be the transgression map on functionals of Definition \ref{transgressionmap}. Then there is a natural surjection 
\begin{equation}
	\mathcal{I}\colon\mathrm{dgMap}_I^0(T[1]I,\mathcal{F}^{(1)}_{CS}) \longrightarrow \mathcal{F}			_{WZ}(\partial M),
\end{equation}
where $\mathrm{Map}^0$ denotes degree-zero maps, and
\begin{align}
	\left[\T \D_{f_{tot}}^{(1)}\right]_{\mathrm{dgMap}_I^0} 
		= \left[\T \D_{f_{min}}^{(1)}\right]_{\mathrm{dgMap}_I^0} & 
		= \mathcal{I}^*S_{gWZ}\\
	\left[\T \D_{f^{1,0}_{tot}}^{(1)}\right]_{\mathrm{dgMap}_I^0} 
		=\left[\T\D_{f_{min}^{1,0}}^{(1)}\right]_{\mathrm{dgMap}_I^0} & 
		= \mathcal{I}^*S_{gWZW}^{1,0},
\end{align}
with polarising functionals $f^\bullet$ as in Definition \ref{Polarisingfunctionals}.
\end{theorem}

\begin{proof}
We begin by explicitly parametrising the space $\mathrm{Map}(T[1]I,\mathcal{F}^{(1)}_{CS})$. We denote by $t$ the coordinate in $I=[0,1]$, and we have maps
\begin{equation*}
	\begin{cases}
		\mathbb{A}=A(t) + a(t)dt\\
		\mathbb{c}=c(t) + \gamma(t) dt\\
		\mathbb{A}^\dag = A^\dag(t) + \beta(t)dt
	\end{cases}
\end{equation*}
where for all $t$, $a(t)\in\Omega^1[-1](M^{(1)})\otimes \mathfrak{g}$ is of degree $-1$, $\gamma(t)\in\Omega^0(M^{(1)})\otimes \mathfrak{g}$ is of degree $0$ and $\beta(t)\in\Omega^2[-2](M^{(1)})\otimes \mathfrak{g}^*$ is of degree $-2$. Restricting to degree-$0$ maps we are left with a parametrisation of $\mathrm{Map}^0(T[1]I,\mathcal{F}^{(1)}_{CS})$ given by the pairs $(A(t),\gamma(t))$, and the defining property of $\mathrm{dgMap}_I^0(T[1]I,\mathcal{F}^{(1)}_{CS})$ (see Definition \ref{Theorem:WZWfromDelta}) is 
\begin{equation}\label{ELcondition}
	\frac{\delta S}{\delta A(t)}=0 \iff d_{A(t)}\gamma(t) = \frac{d}{dt}A(t)
\end{equation} 
Given $\gamma(t)$, equation \eqref{ELcondition} is solved by $A(t)=g(t)^{-1}A_0 g(t) + g(t)^{-1}dg(t)$, with $g(t)\coloneqq \mathrm{Pexp}(\int_0^t \gamma_s ds)$, by Lemma \ref{Lemma:ELtransverse}, and we fix the initial condition to $g_0\equiv g(0)=\mathrm{id}$. Then, we compute the transgression (the total derivative in Eq. \eqref{Deltasemipolarised} vanishes when integrated on $\partial M$, which is by assumption a closed manifold without boundary)
\begin{align*}
	\T \D_{f_{tot}}^{(1)} \equiv \T \D_{f_{min}}^{(1)} & = \T\intl_{\partial M} \frac12 A dc
		= \frac12\intl_{\partial M\times I} \mathbb{A} d\mathbb{c}  \\
	& = \frac12 \intl_{I} dt \intl_{\partial M} A(t)d\gamma(t) + a(t) dc(t).
\end{align*}
Restricting to degree-zero maps we get
\begin{equation}
	\left[\T \D_{f_{tot}}^{(1)}\right]_{\mathrm{dgMap}^0} 
		=\frac12 \intl_{I} dt \intl_{\partial M} A(t)d\gamma(t)
\end{equation}
and by requiring that the pair $(A(t),\gamma(t))$ solve the defining property \eqref{ELcondition} we get, by virtue of Proposition \ref{Proposition:WZderivatives},
\begin{equation}\label{claim}
	\left[\T \D_{f_{tot}}^{(1)}\right]_{\mathrm{dgMap}_I^0} 
		\!\!\!\!=  \frac12 \intl_{I} dt \intl_{\partial M} A_0^{g_t}d\gamma(t) = \intl_{I} dt \frac{d}{dt} S_{gWZ}[A_0, g_t] 
		= S_{gWZ}[A_0,g_1]
\end{equation}

We can define the morphism $\mathcal{I}\colon\mathrm{dgMap}_I^0(T[1]I,\mathcal{F}^{(1)}_{CS}) \longrightarrow \mathcal{F}_{WZ}(\partial M)$ to be 
\begin{equation*}
	\mathcal{I}(A(t),\gamma(t))=(A(0), g(1)),
\end{equation*}
where $g(t)\coloneqq\mathrm{Pexp}(\int_0^t \gamma_s ds)$ is a group valued function $g_t\colon \partial M \longrightarrow G$ for all $t\in I$, and equation \eqref{claim} becomes
\begin{equation}
	\left[\T \D_{f_{tot}}^{(1)}\right]_{\mathrm{dgMap}_I^0}  = \mathcal{I}^*S_{gWZ},
\end{equation}
and since any $g\in G$ can be obtained as the endpoint of a path $g_t=\mathrm{Pexp}(\int_0^t \gamma_s ds)$ for a suitable $\gamma_t$, the map $\mathcal{I}$ is surjective, proving the first statement. 

In the $(1,0)$-case, where we use $f^{1,0}_{tot}$ (cf. Definition \ref{Polarisingfunctionals}), the calculation is formally equivalent, and equation \eqref{ELcondition} implies in particular that $\dot{A}^{1,0}(t) = \partial_{A^{1,0}(t)}\gamma(t)$, which is solved by
$$
	A^{1,0}(t) = g_t^{-1}A(0) g_t + g_t^{-1}\partial g_t
$$ 
where again $g_t=\mathrm{Pexp}(\int_0^t \gamma_s ds)$, and we obtain (Proposition \ref{Proposition:WZderivatives})
\begin{equation*}
	\left[\T \D_{f^{1,0}_{tot}}^{(1)}\right]_{\mathrm{dgMap}_I^0} 
		\!\!\!\!=  \intl_{I} dt \intl_{\partial M} (A_0^{1,0})^{g_t}\delbar\gamma(t) = 
		\intl_{I} dt \frac{d}{dt} S^{1,0}_{gWZW}[A_0^{1,0}, g_t] = S^{1,0}_{gWZW}[A_0^{1,0},g_1],
\end{equation*}
and
\begin{equation*}
\left[\T \D_{f^{1,0}_{tot}}^{(1)}\right]_{\mathrm{dgMap}_I^0}\equiv \left[\T \D_{f_{min}^{1,0}}^{(1)}\right]_{\mathrm{dgMap}_I^0}  = \mathcal{I}^*S^{1,0}_{gWZW}.
\end{equation*}

\end{proof}

\begin{remark}
Theorem \ref{Theorem:WZWfromDelta} shows a constructive way to generate the Wess--Zumino and Wess--Zumino--Witten functionals out of BFV boundary data, a priori without knowledge of what the WZ(W) terms should be. We do however see some dependence on the \emph{choice of polarisation} (in the form of polarising functionals $f^\bullet$), as the WZ(W) functionals are obtained by \emph{``$f$-transforming''} $\bD^\bullet_{CS}$, i.e. they depend on a choice of a representative in the class $[\bD^\bullet_{CS}]_{\LQ-d}$. However, at the level of classical observables, i.e. gauge invariant functionals of the fields, such choice is immaterial. Since a choice of polarisation has a nontrivial effect on the quantum theory, this is hinting at a more general statement at the level of BV theories and quantisation, the CS-WZW relationship we present here being just a leading-term approximation. We believe that this might be related to automorphisms of the quantum theory (canonical transformations) arising from a choice/change of polarisation. See Section \ref{Sec:Polarisationcomments} for more details on this.
\end{remark}

\begin{remark}
When the symmetries of the classical theory come from a Lie-algebra action, the BV formalism can be seen as an extension of the BRST construction (cf. Section \ref{Sec:BV/BRST}). In that case, we can find a polarising functional that makes the $f$-transformed lax BV-BFV theory of BRST type (namely either $f_{tot}$ or $f^{1,0}_{tot}$). This choice of presentation of the BV-BFV data is distinguished and makes the BV-BFV difference  $\bD^\bullet$ into a cocycle for the BRST operator $(\mathcal{L}_{Q_{BRST}} -d)$. Then, Theorem \ref{Theorem:WZWfromDelta} suggests a construction integrating a $(\mathcal{L}_{Q_{BRST}} -d)$-cocycle to a cocycle for the associated group cohomology.
\end{remark}

\begin{remark}
Theorem \ref{Theorem:WZWfromDelta} does not distinguish either between $f_{min}$ and $f_{tot}$ or between $f_{min}^{1,0}$ and $f^{1,0}_{tot}$. However, the procedure outlined here truncates data at codimension 1, as we are integrating along closed codimension-$1$ strata. We expect a distinction to arise from a similar AKSZ transgression procedure from higher dimensional cells (than the one-dimensional interval $I$), a procedure we shall investigate further elsewhere.
\end{remark}

\subsection{A remark on comparing polarisations in CS (and beyond)}\label{Sec:Polarisationcomments}
Geometric quantisations of the phase space $\mathcal{F}^{(1)}_\Sigma$ of a theory on a codimension-1 stratum $\Sigma$ are expected to arrange into a vector bundle (``Friedan-Shenker bundle'') over the space of polarisations $Pol_\Sigma$, with a natural projectively flat connection $\nabla$ allowing one to compare the spaces of states corresponding to infinitesimally close polarisations. The parallel transport of $\nabla$ along a curve on $Pol_\Sigma$ (``BKS\footnote{For Blattner-Kostant-Sternberg.} kernel'' or ``Segal-Bargmann kernel'') gives one a comparison of states in a pair of arbitrary polarisations. 

In the context of Chern-Simons theory, this picture was developed in \cite{ADPW}, for a subspace $ \mathcal{M}^\mr{complex}_\Sigma\subset  Pol_\Sigma$  given by polarisations associated to complex structures on the surface $\Sigma$. In this case, $\nabla$ is the Hitchin connection and the fibre of the bundle\footnote{Here we understand that we are quantising the \emph{reduced} phase space (the moduli space of flat connections on $\Sigma$). Equivalently, (see Section \ref{Sect:Omegacohomology}), we quantise the non-reduced BFV phase space and then pass to the cohomology of the quantum BFV differerential $\Omega$.} is the Verlinde space $H^0_{\bar\partial}(\mathcal{M}^\mr{flat}_\Sigma,\mathcal{L}^{\otimes k})$ , i.e. the space of holomorphic sections, over the moduli space of flat connections on $\Sigma$, of the natural line bundle arising from pushing forward the Noether 1-form, viewed as a connection on a trivial line bundle, along the symplectic reduction (here $k\in\mathbb{Z}$ is the level). The fact that $\nabla$ is only \textit{projectively} flat is an effect related to the nonzero central charge of Wess--Zumino--Witten theory on $\Sigma$.  It is interesting also to consider polarisations not coming from a complex structure on $\Sigma$, e.g. polarisations inferred from a polarisation of the Lie algebra $\g$, see \cite{CMW}. One expects that these can be compared to each other, for different polarisations of $\g$, and also to the ones coming from complex structures on $\Sigma$, by generalised BKS kernels/the parallel transport using generalised Hitchin connection on $Pol_\Sigma$.

In the BV-BFV context, the idea is to realise BKS kernels as partition functions of cylinders $[0,1]\times \Sigma$ (carrying the AKSZ theory obtained from the BFV data on $\Sigma$) with two different polarisations put on the top/bottom of the cylinder.

\subsection{Cohomology of the BFV operator $\Omega$ and gWZW action}\label{Sect:Omegacohomology}
In Chern-Simons theory on a 3-manifold $M$ with boundary $\Sigma$, with phase space $\Omega^\bullet(\Sigma,\g)[1]$ polarised with the base $\Omega^{1,0}(\Sigma,\g)\oplus \Omega^0(\Sigma,\g)[1]\ni  (A^{1,0},c)$, one can consider the quantum BFV operator - the canonical quantisation of the BFV action $S^{(1)}$ on $\Sigma$ (cf. \cite{ABM}): 
\begin{equation}\label{Omega CS}
	\Omega_\Sigma=\int_\Sigma \langle c,\bar\dd A^{1,0}
		- i\hbar\;\dd_{A^{1,0}} \frac{\delta}{\delta A^{1,0}} \rangle 
		- i\hbar\frac12 \langle [c,c],\frac{\delta}{\delta c} \rangle
\end{equation} 
and quantum states $\Psi(A^{1,0},c)\in \mathcal{H}_\Sigma$ annihilated by $\Omega_\Sigma$. Restricting to states of ghost number zero, $\Psi(A^{1,0})$, one can see that the equation 
\begin{equation}\label{Omega Psi =0}
	\Omega_\Sigma \Psi(A^{1,0}) = 0
\end{equation} 
is tantamount\footnote{In this transition we need to integrate by parts in the second term in (\ref{Omega CS}). Here it is important that $\Sigma$ has no boundary.} to $i\hbar \frac{d}{d\epsilon}\big|_{\epsilon=0} \Psi((A^{1,0})^{1+\epsilon \alpha}) + \int_\Sigma \langle \alpha , \bar\dd A^{1,0} \rangle \Psi(A^{1,0}) =0 $ for any $\alpha \in \Omega^0(\Sigma,\g)$, which in turn integrates to
\begin{equation} \label{Psi (Ag) via WZW} 
	\Psi((A^{1,0})^g)= e^{\frac{i}{\hbar} S^{1,0}_{gWZW}(A^{1,0},g)} \Psi(A^{1,0})
\end{equation}
for any $g\in \mr{Map}( \Sigma,G)$.
Thus, the problem of computing the degree-zero cohomology of $\Omega_\Sigma$ acting on states of Chern-Simons theory on the boundary surface $\Sigma$ is naturally related to the gauged Wess--Zumino--Witten theory on $\Sigma$. In \cite{GK} it was proven that the degree zero cohomology of $\Omega_\Sigma$ -- the space of solutions of (\ref{Omega Psi =0}) or, equivalently, of (\ref{Psi (Ag) via WZW}) -- in genus zero case with punctures (corresponding to Wilson lines labeled with representations of $G$ crossing the surface $\Sigma$) coincides with the space of conformal blocks of gWZW theory. Analogous result is expected to hold in arbitrary genus.

In the setting with $n$ punctures $z_1,\ldots,z_k\in \Sigma$ decorated by representations $\rho_1,\ldots, \rho_n$ of $G$, the states are functions of $A^{1,0},c$ with values in $\otimes_{k=1}^n V_k$ (here $V_k$ is the representation space of $\rho_k$), and one needs to add to the r.h.s. of (\ref{Omega CS}) the term $\sum_{k} \rho_k(c(z_k))$. 

Formula (\ref{Psi (Ag) via WZW}) then becomes 
$$
	\Psi((A^{1,0})^g)
		= e^{\frac{i}{\hbar} S^{1,0}_{gWZW}(A^{1,0},g)} \otimes_{k=1}^n \rho_k(g(z_k)) \; \Psi(A^{1,0}).
$$  
It is explained in \cite{ABM} how to obtain this picture for the BFV space of states and the BFV differential in the presence of punctures from an auxiliary BV theory, corresponding to a path integral presentation (Alekseev--Faddeev--Shatashvili formula) for the Wilson lines.

\subsection{Descent equations for Chern--Simons theory}
In this section we discuss solutions for the {\color{red}descent} equations (see Section \ref{Sec. WDE}) as provided by the BV-BFV construction. We stress that the BV formalism encodes data coming from symmetries while at the same time localising to the critical locus of the classical action functional, as opposed to the BRST formalism, whose differential knows about \emph{off-shell} symmetries.

In \cite{Alekseev} a solution was presented of the descent equation for (a representative of) the first Pontrjagin class on a four dimensional manifold with a principal $G$-bundle, $p=<F_A,F_A>$, where $A$ is a principal connection and $F_A$ its curvature. The proposed solution for the descent equation $p=D\omega$, where $D$ is a differential comprising of de-Rham on $M$ and a version of Chevalley-Eilenberg differential, is the inhomogeneous form\footnote{Observe that in this version all fields are of  degree $0$.}
\begin{equation}\label{AlekCocycle}
	\omega=\sum_{i=0}^3\omega_i \equiv 
		\frac12 \left(AdA + dx_1 A+ x_1dx_2\right) + \frac16 A[A,A] - \frac{1}{12}x_1[x_2,x_3],
\end{equation}
with $\omega_i$ the $i-$form part of $\omega$ and $x_i\in \mathfrak{g}$.

A direct interpretation of this comes from the BRST formalism. We denote by $\mathfrak{A}_P$ the space of connections on the principal bundle $P\longrightarrow M$.

\begin{proposition}
Let $Q_{BRST}$ be the Chevalley-Eilenberg operator seen as a vector field on 
$$
	C^\infty(\mathcal{F}_{BRST})\equiv C^\infty(\mathfrak{A}_P \times \Omega^0[1](M,\mathfrak{g})),$$
i.e. $Q_{BRST}(A) = d_Ac$, $Q_{BRST}(c)=\frac12[c,c]$. Then, the following functionals are $(\mathcal{L}_{Q_{BRST}} - d)$-cocycles
\begin{subequations}\label{BRSTCOCYCLES}\begin{align}\label{OurversionofAlekCocycle}
	\mathbb{L}_{BRST}^{\mathrm{I}} & 
		= \frac12 \left(AdA + dc A+ cdc\right) + \frac16 A[A,A] - \frac{1}{12}c[c,c]\\
	\mathbb{L}_{BRST}^{\mathrm{II}} & 
		= \frac12 \left(AdA + c dA\right) + \frac16 A[A,A] - \frac14A[c,c] - \frac{1}{12}c[c,c]
\end{align}\end{subequations}
and their difference is exact: $\mathbb{L}_{BRST}^{\mathrm{I}}-\mathbb{L}_{BRST}^{\mathrm{II}}=\frac12(\mathcal{L}_{Q_{BRST}} - d)(cA)$. 
\end{proposition}
\begin{proof}
This is just a matter of a straightforward computation.
\end{proof}
\begin{remark}
Equation \eqref{AlekCocycle} is reproduced by the cocycle in \eqref{OurversionofAlekCocycle}, by interpreting the terms $cdc$ and $c[c,c]$ with the approriate antisymmetrisation on elements of $\mathfrak{g}$. From the BV-BFV formalism, we have a natural $(\mathcal{L}_{Q} - d)$-cocycle $\bD^\bullet_{CS}$ in $\mathbb{\Omega}^\bullet(\mathcal{L}_{Q} - d)$ given by Equation \eqref{Eq:CSDeltacocycle}. One can now observe that cocycle \eqref{OurversionofAlekCocycle} coincides with $\mathcal{P}_{f_{tot}}\bD^\bullet_{CS}$ of Proposition \ref{CSBRSTtype}, thus realising the proposal of Eq. \eqref{AlekCocycle} in the BV-BFV formalism.
\end{remark}

\begin{remark} It is easy to see that the following are other $(\mathcal{L}_{Q} - d)$-cocycles, all cohomologous to $\mathbb{L}^\bullet_{CS}$ in $\mathbb{\Omega}^\bullet(\mathcal{L}_{Q} - d)$:
\begin{subequations}\begin{align}
	L_{BV}^{a,\mathrm{I}}&=\frac12 \left(AdA  +  dc A + cdc \right) + dA^\dag c \\\notag
		&+ [A,A^\dag]c +\frac16 A[A,A] + \frac12 c^\dag[c,c] - \frac{1}{12}c [c,c];\\
	L_{BV}^{b,\mathrm{I}} &
		= \frac12\left( AdA  + 3 Adc + 3 cdc + 2 A^\dag dc\right) + [A,A^\dag]c + \frac16 A[A,A] 
		- \frac{1}{12} c[c,c]\\\notag
		& + \frac12 \left(c^\dag [c,c] + c[A,A] + A^\dag[c,c] + A[c,c]\right); 
\end{align}\end{subequations}
where we explicitly parametrise $\mathcal{A}=(c,A,A^\dag,c^\dag)$. Moreover, by realising $\mathcal{F}_{BRST}$ as the zero-section in $\mathcal{F}_{CS}=T^*[-1]\mathcal{F}_{BRST}$ (i.e. defined by $A^\dag=c^\dag=0$) we have the following relations:
\begin{subequations}\begin{align}
	L_{BV}^{a,\mathrm{I}}\Big\vert_{A^\dag = c^\dag =0} &= \mathbb{L}_{BRST}^{\mathrm{I}}\\
	L_{BV}^{b,\mathrm{I}}\Big\vert_{A^\dag = c^\dag =0} &= \mathbb{L}_{BRST}^{\mathrm{I}} 
		+ cF_A + (d-\mathcal{L}_Q)(cA) \approx \mathbb{L}_{BRST}^{\mathrm{II}} 
		+ (d-\mathcal{L}_Q)(cA)
\end{align}\end{subequations}
where the symbol $\approx$ means that the equivalence is up to classical equations of motion, i.e. when $F_A=0$.
\end{remark}

\subsection{Abelian Chern Simons theory on Lorentzian manifolds}

In this section we focus on Abelian Chern-Simons theory on a Lorentzian manifold $(M,g)$. The reason that we consider it separately, and despite the more general picture outlined previously, is for its applications in condensed matter physics, where it is used as an effective theory e.g. in the description of the fractional quantum Hall effect (FQHE) \cite{FB,F1}. Moreover we shall use this example to recall how polarisations can be obtained from Lorentzian metrics. 

The gauge group we shall consider is $G = U(1)^N$. Let us choose a nondegenerate pairing $\langle\cdot,\cdot\rangle$ on $\RR^N$ (interpreted as the Lie algebra of $U(1)^N$). The data of the theory are given by 
\begin{equation}
\mathfrak{F}=(\Omega^{\bullet}(M,\RR^N), \alpha^{\bullet}_{CS},L^\bullet_{CS},Q_{CS})\\
\end{equation}
where equations \eqref{CSdata} simplify to 
\begin{align*}
\theta^{\bullet}_{CS} &= \frac12 \cA \delta \cA \\
L^{\bullet}_{CS}[\mathcal{A}] &= \frac12 \mathcal{A}d \mathcal{A}\\
Q[\mathcal{A}] &= d\mathcal{A}.
\end{align*}
Consider a 3-manifold $M$ with boundary $\de M$ with a metric $g$ such that both $M$ and $\de M$ have Lorentzian signature. As an example, consider a subset of standard Minkowski space of the form $\Lambda = \RR \times \Omega$ with boundary $\de\Lambda = \RR \times \de \Omega$. 
\begin{proposition}
The Lorentzian metric $g$ induces a splitting of the space of 1-forms on the boundary
\begin{equation}
\Omega^1(\de M) = \Omega^1_+(\de M) \oplus \Omega^1_-(\de M)
\end{equation}
such that $\Omega^1_\pm(\de M) \subset \Omega^1(\de M)$ are Lagrangian with respect to the symplectic form on $\Omega^1(\de M)$ given by  
\begin{equation}
\omega(A,B) = \int_{\de M} A \wedge B
\end{equation}
\end{proposition}
\begin{proof}
The important thing to note is that the Hodge $\star$-operator squares to $+1$ on $\Omega^1(\de M)$ for a Lorentzian metric. One can then let $\Omega^1_\pm(\de M)$ be the $\pm 1$-eigenspaces of $\star$. This is a decomposition into Lagrangian subspaces. 
\end{proof}
In a Minkowski plane, this is the splitting into\footnote{Notice that $\star dx_\pm = \mp dx_\pm$.} $dx_+ = dx + dt$ and $dx_- = dx - dt$ components of a 1-form. Then, denoting by 
\begin{equation}
F^+ = \intl_{\partial M} \frac12(cA^+ + A_+A_-)
\end{equation}
the (integrated) polarising functional, and repeating the analysis of the non-abelian case from Theorem \ref{Theorem:WZWfromDelta}, we conclude that 
\begin{equation}
S_{tot}^+[\mathcal{A}] = S_{CS}[\mathcal{A}] + F^+ + S^+_{gWZW}[g,A]
\end{equation}
is invariant under finite gauge transformations, and can be obtained by a transgression procedure for the $f$-transformed BV-BFV difference $\D^{(1)} - \LQ F^+$. In degree 0, we recover - as a functional of boundary fields correcting the failure of bulk gauge invariance -  the action functional of chiral $U(1)$ currents (\cite[Section 5]{FB})
\begin{equation}
S[g,A] := S_{tot}^+[\mathcal{A}] - S_{CS}[\mathcal{A}] = \int_{\partial M}  A_+\partial_-\phi + \frac12 A_+A_- + \frac12 \partial_+\phi \partial_-\phi\label{eq:bdryaction}
\end{equation}
where $g = \exp(i\phi)$. Formally integrating over the field $\phi$, we obtain the effective \emph{edge} action 
\begin{equation}
\Gamma[A] = \int_{\de M}A_+A_- + A_+\frac{\partial_-^2}{\square}A_+.
\end{equation}
The fact that the gauge anomalies of $\Gamma[A]$ and $S_{CS}[A]$ cancel has been interpreted as an instance of holography \cite{F1}. This is further evidence that the transgression procedure outlined in Definition \ref{transgressionmap}, following Theorem \ref{Theorem:WZWfromDelta}, produces holographic counterparts on the boundary of  theories defined in the bulk. We stress that $S[g,A]$ and $\Gamma[A]$ (up to gauge invariant terms) are uniquely determined from $S_{CS}$. We conclude that in the case of Chern-Simons theory, the sum of the polarising functional and the transgression of the BV-BFV difference\footnote{Restricted to the transversal EL locus of Definition \ref{Theorem:WZWfromDelta}.} generate the unique ``boundary action functional'' \eqref{eq:bdryaction}, eliminating the gauge anomaly. We believe that this holds true in greater generality.

\section{BF Theory}\label{BFTheory}
In this section we analyse BF theory (see e.g. \cite{CCFM, CR, Mn1}) from the point of view of a fully-extended BF-BFV theory. After describing the general construction for $m$ space-time dimensions, we will focus on $m=3$.

We will discuss how one can construct an action functional analogous to Wess--Zumino(--Witten), denoted $S_{\tau F}$, arising as the failure of BF theory under finite gauge transformations (in three spacetime dimensions), similarly to Lemma \ref{Lemma:WZWfailure}. By choosing appropriate polarising functionals (cf. Definition \ref{polarisationdef}) we will show how the BV-BFV diffrences $\bD^\bullet$ for BF theory can be related to $S_{\tau F}$, in a completely analogous fashion to Theorem \ref{Theorem:WZWfromDelta}. Futhermore, by choosing a complex structure on a 2-dimensional stratum, we can relate $S_{\tau F}$ to two copies of gauged Wess--Zumino--Witten functional (cf. Proposition \ref{DoubleWZW}), an explanation of which is given by observing that BF theory can be made equivalent to Chern--Simons theory for a specific choice of structure group (see Theorem \ref{Theorem:BFtoCS}).

Finally, in Section \ref{Sect:BRSTBF} we will show how BF theory can be put in BRST form, similarly to what was done for Chern--Simons theory in Proposition \ref{CSBRSTtype}.

\begin{definition}\label{BVBFtheory}
\emph{Lax BF theory} on the $m$-dimensional manifold $M$ is defined to be the lax BV-BFV theory $\mathfrak{F}_{BF}=\left(\mathcal{F}_{BF} , \theta_{BF}^\bullet, L_{BF}^\bullet, Q\right)$ for the space of fields 
\begin{equation}
	\mathcal{F}_{BF}\coloneqq\Omega^\bullet[1 - \bullet](M,\mathfrak{g})
		\times \Omega^\bullet[m-2 -\bullet](M,\mathfrak{g}^*)  \ni (\cA,\cB)
\end{equation} 
with Lagrangian functional given by 
\begin{equation}
	L^\bullet_{BF}\coloneqq \langle\cB, F_{\cA}\rangle,
\end{equation} 
where $F_{\cA} = d\cA + \frac12 [\cA,\cA]$ and $\langle\cdot,\cdot\rangle$ denoting the natural pairing between $\mathfrak{g}$ and its dual; the one form 
\begin{equation}
	\theta_{BF}^\bullet \coloneqq \langle\cB, \delta \cA\rangle
\end{equation} 
and the cohomological vector field
\begin{equation}
Q\cA = F_{\cA}; \ \ \ Q\cB=d_\cA\cB.
\end{equation}
\end{definition}

Lax BF theory admits a full strictification, following what was presented in \cite{CMR}. The construction is almost identical to Theorem \ref{StrictCSTheorem}:
\begin{propdef}[\cite{CMR}]
The strict BV-BFV codimension-$k$ data associated with lax BF theory and a codimension-$k$ stratum $M^{(k)}\subset M$ is given by 
$$
	\mathcal{F}_{BF}^{(k)}=\Omega^\bullet[1 - \bullet](M^{(k)},\mathfrak{g})
		\times \Omega^\bullet[m-2 -\bullet](M^{(k)},\mathfrak{g}^*), 
$$ 
together with the codimension-$1$ functional and one-form
$$
	S_{BF}^{(k)} = \intl_{M^{(k)}}[L_{BF}]^{m-k};\ \ \alpha^{(k)}_{BF} = \intl_{M^{(k)}}[\theta_{BF}]^{m-k}.
$$
\end{propdef}

The first deviation we observe between Chern--Simons theory and BF theory is related to Theorem \ref{Deltacocycle}.

\begin{proposition}
The BV-BFV difference $\bD^\bullet$ vanishes for all $0\leq k \leq m$. Hence, the total Lagrangian $\mathbb{L}_{BF}^\bullet\equiv L_{BF}^\bullet$ is an $(\LQ - d)$-cocycle. 
\end{proposition}

\begin{proof}
The first statement is a consequence of the simple calculation: 
$$
	\iota_Q \alpha_{BF}^\bullet = \iota_Q (\cB\delta\cA) = \cB F_{\cA}
		= L_{BF}^{\bullet},
$$
since $\bD_{BF}^\bullet = L_{BF}^\bullet - \iota_Q\alpha_{BF}^\bullet$.
Then, $L^\bullet_{BF}$ coincides with the total Lagrangian (cf.Eq. \eqref{totalfunctional}), and is therefore an $(\LQ -d)$-cocycle. More directly, imposing Bianchi identities on $F_{\cA}$, we have $\LQ[\cB F_{\cA}]^{m-k} = [d_{\cA}\cB F_{\cA}]^{m-k} = d[\cB F_{\cA}]^{(m-k-1)}$. 
\end{proof}

\subsection{Three-dimensional BF theory}
Let us now specify our discussion to BF theory on a three-dimensional manifold $M$. The general BV structure given in Definition \ref{BVBFtheory} encodes the infinitesimal symmetry of the classical (i.e. degree-$0$) BF action
$$
	S^{cl}_{BF}\coloneqq\intl_{M} \langle B, F_A\rangle
$$
for $(A,B)\in F_{BF}^{cl}\coloneqq\Omega^1(M,\mathfrak{g})\times\Omega^1(M,\mathfrak{g^*})$, generated by the transformations
\begin{subequations}\begin{align}
	A & \longmapsto A + d_Ac\\
	B & \longmapsto B + [c,B] + d_A\tau
\end{align}\end{subequations}
where $c\in\Omega^0(M,\mathfrak{g})$ and $\tau\in\Omega^0(M,\mathfrak{g}^*)$. If we consider $\Omega^0(M,\mathfrak{g})\ltimes \Omega^0(M,\mathfrak{g}^*)$ as a Lie algebra, together with the pointwise adjoint action on $\mathfrak{g}$ on itself, we gather that there is a Lie group integrating it, as follows (see e.g. \cite{Mn1, CR}).

\begin{definition}\label{BFgaugegroup}
Consider the semi-direct product of a Lie group $G$ with the dual of its Lie algebra seen as an Abelian group. The associated gauge group is given by $\mathcal{G}\coloneqq \Omega^0(M,G) \ltimes \Omega^0(M,\mathfrak{g}^*)\ni(g,\tau)$, with (pointwise) product structure
$$
	(h, \tau')\cdot (g,\tau) = (hg, (\tau')^g + \tau ) = (hg, g^{-1}\tau'g + \tau)
$$
and the (right) action on fields $(A,B)\in\Omega^1(M,\mathfrak{g})\times \Omega^1(M,\mathfrak{g}^*)$ reads:
\begin{equation}
	(A,B)^{(g,\tau)} \coloneqq (A^g, B^g + d_{A^g} \tau)
\end{equation}
with $A^g \coloneqq g^{-1} A g + g^{-1} dg$ and $B^g = g^{-1}Bg$.
\end{definition}

\begin{propdef}\label{propdefBFWZW}
Consider a three-dimensional manifold with boundary $(M,\partial M)$, the space of fields $F_{\tau F}(\partial M)\coloneqq \Omega^0(\partial M,G\ltimes\mathfrak{g}^*)\times\Omega^1(\partial M,\mathfrak{g})$ and a functional over it:
\begin{equation}
	S_{\tau F}[(g, \tau), A]= \intl_{\partial M} \langle \tau^{g^{-1}}, F_A\rangle,
\end{equation}
where $\tau^{g^{-1}} = g \tau g^{-1}$. Then
\begin{equation}
	S_{BF}^{cl}\left[(A,B)^{(g,\tau)}\right] - S_{BF}^{cl}[(A,B)] = S_{\tau F}[(g,\tau),A].
\end{equation}
Moreover, 
\begin{equation}\label{STgaugecorrection}
	S_{\tau F}[(g,\tau)^{-1}\cdot(h,\chi), A^g] = S_{\tau F}[(h,\chi),A] - S_{\tau F}[(g,\tau),A].
\end{equation}
\end{propdef}

\begin{proof}
To check the first statement, using the Bianchi identity for the transformed connection ${}^g\!A$, it is sufficient to compute:
\begin{multline*}
	S^{cl}_{BF}\left[(A,B)^{(g,\tau)}\right] = \intl_{M} \langle g^{-1}Bg, F_{A^g}\rangle 
		+ \langle d_{A^g}\tau, F_{A^g}\rangle = \intl_{M}\langle B,F_A\rangle 
		+ \intl_{M}d\langle \tau, F_{A^g} \rangle.
\end{multline*}
The second statement instead comes from the simple observation that $(g,\tau)^{-1}\cdot (h,\chi) = (g^{-1}h, \chi -\tau^{g^{-1}h})$ and 
\begin{equation*}
	S_{\tau F}[(g^{-1}h, \chi -\tau^{g^{-1}h}), A^g] = \intl_{\partial M} \langle
		\left(\chi-\tau^{g^{-1}h}\right)^{h^{-1}g}, F_{A^g}\rangle 
		= \intl_{\partial M} \langle \chi^{h^{-1}}, F_{A}\rangle - \langle \tau^{g^{-1}}, F_{A}\rangle .
\end{equation*}
\end{proof}

\begin{remark}
We can interpret the functional $S_{\tau F}$ as an analogue of what the (gauged)-Wess--Zumino action functional is for Chern--Simons theory. Indeed, it encodes the failure of gauge invariance of the classical BF functional under the action of the gauge group of Definition \ref{BFgaugegroup}, in the presence of boundaries, in a similar way to gWZ. In fact, Equation \eqref{STgaugecorrection} tells us that the sum $S_{BF}^{cl} + S_{\tau F}$ is invariant under finite gauge transformations. 
\end{remark}

\begin{proposition}\label{semidirectPexp}
Let $(g_t, \tau_t)$ be a path in $\mathcal{G}$, such that $g_t = \mathrm{Pexp}(\int_0^t \gamma_s ds)$ for $\gamma_t$ a path in $\Omega^0(M,\mathfrak{g})$ and 
$$
	\tau_t = g_t^{-1} \left(\int\limits_0^t \beta_{t'}^{g_{t'}^{-1}} dt'\right) g_t
$$
a path in $\Omega^0(M,\mathfrak{g}^*)$, with $\beta_{t}\in\Omega^0(M,\mathfrak{g}^*)$ for all $t$. Then, denoting by
\begin{equation*}
	(A(t),B(t))\coloneqq (A,B)^{(g_t,\tau_t)} = \left( A^{g_t}, B^{(g_t,\tau_t)}\right),
\end{equation*}
with $B^{(g_t,\tau_t)}\coloneqq g_t^{-1} Bg_t + d_{A^{g_t}}\tau_t$, and $A^{g_t}= g^{-1}_tA g_t + g^{-1}dg_t$, we have that 
\begin{equation}\label{transversalequationsgaugeBF}
	\begin{cases}
		\dot{A}(t) = d_{A(t)}\gamma_t\\
		\dot{B}(t) = - [\gamma_t, B(t)] + d_{A(t)}\beta_t.
	\end{cases}
\end{equation}
\end{proposition}

\begin{proof}
First we observe that equation $\dot{A}(t) = d_{A(t)}\gamma_t$ follows from Lemma \ref{Lemma:ELtransverse}. Then, the second of Equations \ref{transversalequationsgaugeBF} is a matter of straightforward computations: recalling that $\dot{g_t} = g_t \gamma_t$ we have
\begin{multline*}
\frac{d}{dt}(g_t^{-1} Bg_t + d_{A^{g_t}}\tau_t) = -g_t^{-1}\dot{g}_t g_t^{-1} Bg_t + g_t^{-1}B \dot{g}_t  + d_{A^{g_t}}\dot{\tau}_t + \left[\frac{d}{dt} A^{g_t},\tau_t\right]\\
= -[\gamma_t , g_t^{-1} Bg_t] + d_{A^{g_t}} \left(-\left[\gamma_t, \tau_t\right] + g_t^{-1} \frac{d}{dt}\int\limits_0^t \beta_{t'}^{g_{t'}^{-1}} dt' g_t\right) + [d_{A^{g_t}} \gamma_t,\tau_t]\\
= -[\gamma_t , g_t^{-1} Bg_t] - [d_{A^{g_t}}\gamma_t,\tau_t] - [\gamma_t,  d_{A^{g_t}}\tau_t] + d_{A^{g_t}}\beta_t + [d_{A^{g_t}} \gamma_t,\tau_t]\\
= -[\gamma_t, g_t^{-1} Bg_t + d_{A^{g_t}}\tau_t] + d_{A^{g_t}}\beta_t = -[\gamma_t, B(t)] + d_{A(t)}\beta_t.
\end{multline*}
\end{proof}

\begin{lemma}\label{lemmaderivativeST}
Consider the path $(g_t,\tau_t)$ defined in Proposition \ref{semidirectPexp}. Then, we have
\begin{equation}
\frac{d}{dt} S_{\tau F}[(g_t,\tau_t),A] = S_{\tau F}[(g_t, \beta_t),A].
\end{equation}
\end{lemma}

\begin{proof}
This is a simple calculation:
\begin{multline*}
\frac{d}{dt} S_{\tau F}[(g_t,\tau_t),A] = \intl_{\partial M}\frac{d}{dt} \left(g_t\left( g^{-1}_t \intl_0^t \beta_{t'}^{g_{t'}^{-1}} dt'\ g_t\right) g_t^{-1}\right)F_A = \intl_{\partial M} \beta_t^{g^{-1}_t} F_A.
\end{multline*}
\end{proof}

\begin{remark}
We observe that Equations \eqref{transversalequationsgaugeBF} coincide with the transversal Euler--Lagrange equations for BF theory on a cylinder. As a matter of fact, on the three-dimensional cylinder $M=\partial M \times \mathbb{R}$, splitting the fields as $A=A_\perp dt + A^\partial$ and $B=B_\perp dt +B^\partial$, we gather that the equations of motion $F_A=0$ and $d_AB=0$ also split as
\begin{align*}
	-\de_tA^\partial + d_{A^\partial}A_\perp = 0\\
	\de_tB^\partial + d_{A^\partial}B_\perp  = 0\\
	F_{A^\partial} = 0 \\
	d_{A^\partial} B^\partial = 0.
\end{align*}
The first two equations are \emph{evolution equations}, i.e. transversal to $\partial M$ along the $\mathbb{R}$-direction, and they are solved by $(A^\partial(t),B^\partial(t))=(g_t,\tau_t)\cdot (A^\partial(0),B^\partial(0))$, where $(g_t,\tau_t)$ are defined as in Proposition \ref{semidirectPexp}, with $\gamma_t=A_\perp$ and $\beta_t = B_\perp$.
\end{remark}

We want to turn our attention now to lax BF theory, as presented in Definition \ref{BVBFtheory}. We know that the BV-BFV differences vanish, $\bD^\bullet\equiv 0$, however, we can still choose a nontrivial boundary term $f^\bullet$. As a matter of fact, on a manifold with boundary $(M,\partial M)$ we can pick $f^\bullet = f_{min}\coloneqq  \tau B^\dag$ and compute (cf. Proposition \ref{PolchangeDeltaL})
\begin{equation}\label{poldeltaBF}
	\D_{f_{min}}^{(1)}\equiv  \D^{(1)} - \intl_{\partial M}\LQ f_{min} = - \intl_{\partial M}\LQ f_{min} = \intl_{\partial M}\tau F_A
\end{equation}

Following this construction, we can now state the main result in this section:

\begin{theorem}\label{nonabelianBFtheorem}
Consider BF theory on a manifold with boundary $(M,\partial M)$ for a connected, simply connected structure group $G$. Let $\mathrm{Map}(T[1]I, \mathcal{F}^{(1)}_{BF})$ be the AKSZ space of fields with target the strict BFV theory $(\mathcal{F}_{BF}^{(1)}, S^{(1)}_{BF}, \Omega^{(1)}_{BF}, Q^{(1)}_{BF})$ and let $\mathbb{T}$ be the transgression map on functionals of Definition \ref{transgressionmap}. Then, there is a natural surjection
\begin{equation}
	\mathcal{I} \colon \mathrm{dgMap}^0_I(T[1]I, \mathcal{F}^{(1)}_{BF}) 
		\longrightarrow F_{\tau F}(\partial M)
\end{equation}
and
\begin{equation}
	\left[\mathbb{T} \D^{(1)}_{f_{min}}\right]_{\mathrm{dgMap}^0_I} = \mathcal{I}^*S_{\tau F}
\end{equation}
with $f_{min}= \tau B^\dag$.
\end{theorem}

\begin{proof}
We start by parametrizing the space of AKSZ fields by 
\begin{align*}
	\mathbb{A} & = A(t) + a(t)dt\\
	\mathbb{c} & = c(t) + \gamma(t) dt\\
	\mathbb{t} & = \tau(t) + j(t) dt\\
	\mathbb{B} & = B(t) + b(t) dt
\end{align*}
and similarly for the antifields, although since we are interested in maps of degree-$0$ we can neglect them in what follows. Observe that $A(t), \gamma(t), j(t)$ and $B(t)$ are the only maps of degree-$0$. Then, the transgression of $\D^{(1)}_{f_{min}}$, referring to Equation \eqref{poldeltaBF}, reads
$$
\left[\mathbb{T}\D^{(1)}_{f_{min}}\right]_{\mathrm{Map}^0} = \left[\intl_{I\times\partial M} \mathbb{t} F_{\mathbb{A}}\right]_{\mathrm{Map}^0} = \intl_{I} dt \intl_{\partial M} \langle j(t), F_{A(t)}\rangle.
$$
Now, the restriction to $\mathrm{dgMap}^0_I(T[1](I),\mathcal{F}^{(1)}_{BF})$ enforces the following equations
\begin{align*}
	\dot{B}(t) & = - [\gamma(t),B(t)] + d_{A(t)} j(t)\\
	\dot{A}(t) & = d_{A(t)}\gamma(t)
\end{align*}
which are solved by $(A(t),B(t))= (g_t,\tau_t)\cdot(A(0) B(0))$ with $g_t=\mathrm{Pexp}(\gamma(t))$ and $\tau_t = g_t^{-1}\left[\int_0^t g_{t'} j(t) g_{t'}^{-1} dt' \right]g_t$. Then, from Lemma \ref{lemmaderivativeST}, and with $I=[0,1]$ we get
\begin{equation*}
	\left[\mathbb{T}\D^{(1)}_{f_{min}}\right]_{\mathrm{dgMap}_I^0}
		=\intl_{I} dt \intl_{\partial M}\langle{}j(t)^{g_t^{-1}} F_{A(0)}\rangle
		=\intl_{I} dt \intl_{\partial M}\frac{d}{dt}\langle \tau_t^{g_t^{-1}} , F_A \rangle = S_{\tau F}[(g_1,\tau_1),A],
\end{equation*}
which, upon defining the surjection $\mathcal{I}\colon (g_t,\tau_t,A(t),B(t)) \longmapsto (g_1,\tau_1, A(0))$, allows us to conclude the proof.
\end{proof}

\begin{remark}
The choice of $f_{min} = \tau B^\dag$ induces the shift in the one-form 
$$
	\mathcal{P}_{f_{min}}(\theta^\bullet) = B\delta A + A^\dag \delta c + B^\dag \delta \tau 
		+ \text{higher codimension}.
$$ 
This is compatible with choosing a polarisation in $\mathcal{F}^{(1)}_{BF}$ whose space of leaves is parametrised by fields $(A,c,\tau)$.
\end{remark}

\subsection{Complex-structure polarisations for 3d BF theory}
In this section we will focus on BF theory on a three-dimensional manifold with boundary, but we will consider a boundary term that uses a complex structure on the boundary surface to pair fields.

As a matter of fact, as already observed for Chern--Simons theory in Section \ref{boundarypolarisation}, the choice of a complex structure on the 2 dimensional boundary surface defines a splitting in the space of boundary fields $\mathcal{F}^{(1)}$, as we can write $B\vert_{M^{(1)}} = B^{1,0} + B^{0,1}$ and $A\vert_{M^{(1)}} = A^{1,0} + A^{0,1}$. We can then add an $(\LQ -d)$-coboundary to $\mathbb{L}^\bullet_{BF}$, by defining $f^{1,0}_{BF}=B^{1,0}A^{0,1} + \tau B^\dag$, as follows:
\begin{equation}
	\mathcal{P}_{f^{1,0}_{BF}}(\mathbb{L}^\bullet)_{BF}
		= \mathbb{L}^\bullet_{BF} - (\LQ - d)(B^{1,0}A^{0,1} + \tau B^\dag).
\end{equation}
As argued in Section \ref{boundarypolarisation}, this is equivalent to adding a $d-$exact term to the top-form Lagrangian $L_{BF}^{(0)}$ and a $(\LQ-d)$-exact term to $\bD_{BF}\equiv 0$, so that 
\begin{align}\label{BFpolarisingcoboundary}
	\mathcal{P}_{f^{1,0}_{BF}}(\bD^\bullet)_{BF}& =(d-\LQ) (B^{1,0}A^{0,1} + \tau B^\dag) \\\notag
	& = d(B^{1,0}A^{0,1} + \tau B^\dag) - \del\tau A^{0,1} - \tau [A^{1,0},A^{0,1}] 
		+ B^{1,0}\delbar c + \tau F_A \\\notag
	& = d(B^{1,0}A^{0,1} + \tau B^\dag) -d(\tau A) + d\tau A 
		- \del\tau A^{0,1} + B^{1,0}\delbar c \\\notag
	& = d(B^{1,0}A^{0,1} + \tau B^\dag) + A^{1,0} \delbar \tau + B^{1,0}\delbar c - d(\tau A).
\end{align}

\begin{propdef}\label{DoubleWZW}
We define \emph{polarised BF theory} the classical functional obtained by the choice of a complex structure on $\partial M$, as follows
\begin{equation}
	S_{BF}^{1,0}[(A,B)] = \intl_{M} \langle B, F_A\rangle 
		+ \intl_{\partial M} \langle B^{1,0},A^{0,1} \rangle.
\end{equation}
Moreover, considering again the space of fields $F_{\tau F}(\partial M)\ni (g,\tau)$ of Definition \ref{propdefBFWZW}, we will call \emph{gauged, split Wess--Zumino--Witten functional} the following expression
\begin{equation}
	S_{\tau F}^{1,0}[A,B,g,\tau]\coloneqq\intl_{\partial M}\langle(A^{1,0})^g, \delbar \tau\rangle 
		+ \langle g^{-1} B^{1,0} g, (A^{0,1})^g\rangle.
\end{equation}
Then, we have the following: denoting by $(A,B)^{(g,\tau)}$ the action of the group $\widetilde{G}$ on the fields,
\begin{equation}\label{Eq:BFWZW}
	S_{BF}^{1,0}[(A,B)^{(g,\tau)}] - S_{BF}^{1,0}[(A,B)] = S_{\tau F}^{1,0}[A,B,g,\tau].
\end{equation}
Finally, if we consider a path $(g_t,\tau_t)$ in $\widetilde{G}$ as in Proposition \ref{semidirectPexp}, we obtain:
\begin{equation}\label{InfinitesimaldoubleWZW}
\frac{d}{dt} S_{\tau F}^{1,0}[g_t,\tau_t,A,B] = \intl_{\partial M} (B^{1,0})^{(g_t,\tau_t)}\delbar \gamma_t + (A^{1,0})^{g_t} \delbar \beta_t.
\end{equation}
\end{propdef}

\begin{proof}
Using Proposition/Definition \ref{propdefBFWZW} we get 
\begin{multline*}
	S_{BF}^{1,0}[(A,B)^{(g,\tau)}] - S_{BF}^{1,0}[(A,B)]= 
		\intl_{M}d\left[ \langle\tau, F_{A^g}\rangle + \langle g^{-1} B^{1,0} g, (A^{0,1})^g\rangle 
		+ \langle\del_{(A^{1,0})^g}\tau, (A^{0,1})^g\rangle\right]\\
	\intl_{\partial M} d \langle\tau, A^g\rangle - \langle d\tau, A^g\rangle 
		+ \langle g^{-1}B^{1,0} g, (A^{0,1})^g\rangle + \langle \del\tau, (A^{0,1})^g\rangle \\
	= \intl_{\partial M} \langle g^{-1}B^{1,0} g, (A^{0,1})^g\rangle 
		- \langle \delbar\tau, (A^{1,0})^g\rangle = S_{\tau F}^{1,0}[A,B,g,\tau],
\end{multline*}
proving the first statement in Eq. \eqref{Eq:BFWZW}. Moreover, we compute
\begin{multline}
	\frac{d}{dt}S_{\tau F}^{1,0}[g_t,\tau_t,A,B] = \intl_{\partial M} \langle\del_{(A^{1,0})^{g_t}}\gamma_t, \delbar\tau_t\rangle 
		+ \langle(A^{1,0})^{g_t},\delbar\left(\beta _t - [\gamma_t,\tau_t]\right)\rangle \\
	- \langle[\gamma_t,g_t^{-1}B^{1,0}g_t],(A^{0,1})^{g_t}\rangle 
		+ \langle g_t^{-1}B^{1,0}g_t,\delbar_{(A^{0,1})^{g_t}}\gamma_t\rangle\\
	=\intl_{\partial M}\langle\del\tau_t,\delbar\gamma_t\rangle + \langle[(A^{1,0})^{g_t},\tau_t],\delbar\gamma_t\rangle 
		+ \langle (A^{1,0})^{g_t},\delbar\beta_t\rangle 
		+ \langle g_t^{-1}B^{1,0}g_t,\delbar\gamma_t\rangle\\
	= \intl_{\partial M}\langle \left(g_t^{-1}B^{1,0}g_t 
		+ \del_{(A^{1,0})^{g_t}}\tau_t\right), \delbar\gamma_t\rangle 
		+ \langle(A^{1,0})^{g_t},\delbar\beta_t\rangle\\
	= \intl_{\partial M} (B^{1,0})^{(g_t,\tau_t)}\delbar \gamma_t + (A^{1,0})^{g_t} \delbar \beta_t.	
\end{multline}
\end{proof}

The gauge failure of the polarised BF action is then controlled by the polarisation of the BV-BFV difference, in the same way of Theorem \ref{nonabelianBFtheorem}:

\begin{theorem}
With the same assumptions of Theorem \ref{nonabelianBFtheorem}, we have that
\begin{equation}
	\left[\T\D^{(1)}_{f^{1,0}_{BF}}\right]_{\mathrm{dgMap}^0_I} = \mathcal{I}^*S_{\tau F}^{1,0}
\end{equation}
where 
$$
	\mathcal{I}\colon \mathrm{dgMap}(T[1]I,\mathcal{F}^{(1)}) \longrightarrow F_{\tau F}(\partial M)
$$
and $f^{1,0}_{BF}=B^{1,0}A^{0,1} + \tau B^\dag$.
\end{theorem}

\begin{proof}
From Equation \eqref{BFpolarisingcoboundary} it is easy to gather that the polarised $1$-difference on a manifold with boundary $(M,\partial M)$, and $\partial\partial M=\emptyset$, reads
\begin{equation*}
	\D^{(1)}_{f^{1,0}_{BF}} = \intl_{\partial M} A^{1,0} \delbar \tau + B^{1,0}\delbar c.
\end{equation*}
Then, with the same parametrisation of the space of AKSZ fields as in Theorem \ref{nonabelianBFtheorem}, we get, in degree-zero
\begin{equation*}
	\left[\T\D^{(1)}_{f^{1,0}_{BF}}\right]_{\mathrm{Map}^0} = \intl_{I}dt \intl_{\partial M} A^{1,0}\delbar j(t) 
		+ B^{1,0} \delbar \gamma(t)
\end{equation*}
Then, it is easy to gather that, using The results in Proposition/Definition \ref{DoubleWZW}, especially Equation \eqref{InfinitesimaldoubleWZW}, since on $\mathrm{dgMap}^0_I$ maps have to satisfy
\begin{align*}
	\dot{B}(t) & = - [\gamma(t),B(t)] + d_{A(t)} j(t)\\
	\dot{A}(t) & = d_{A(t)}\gamma(t),
\end{align*}
and these equations are solved by 
$$
	(A(t),B(t)) = (A(0)^{g_t}, B(0)^{(g_t,\tau_t)} ) \equiv (g_t^{-1}A(0)g_t + g_t^{-1}dg_t, g_t^{-1}B(0)g_t + d_{A^{g_t}}\tau_t)
$$ 
with $g_t=\mathrm{Pexp}(\gamma(t))$ and $\tau_t = g_t^{-1}\left[\int_0^t g_{t'} j(t) g_{t'}^{-1} dt' \right]g_t$, one has
\begin{align*}
	\left[\T\D^{(1)}_{f^{1,0}_{BF}}\right]_{\mathrm{dgMap}_I^0} 
		& = \intl_{I}dt \intl_{\partial M} (A(0)^{1,0})^{g_t} \delbar j(t)
			+ (B^{1,0})^{(g_t,\tau_t)} \delbar \gamma(t) \\
		& = \intl_{I}dt \frac{d}{dt}S_{gT}[g_t,\tau_t,A,B].
\end{align*}
Upon defining $\mathcal{I} \colon \mathrm{dgMap}^0_I(T[1]I,\mathcal{F}^{(1)}) \longrightarrow F_{\tau F}(\partial M)$ as 
$$
	(A(t),B(t),\gamma(t),\tau(t)) \longmapsto (A(0)^{g_1}, B(0)^{(g_1,\tau_1)}),
$$
we conclude the proof.
\end{proof}

A direct explanation of this result comes from the following observation, that for a cotangent Lie algebra\footnote{A cotangent Lie algebra is of the form $\mathfrak{g}=T^*\mathfrak{h} =\mathfrak{h} \ltimes \mathfrak{h}^*$.}, Chern--Simons theory can be written as an $f$-transformed BF theory.

\begin{theorem}\label{Theorem:BFtoCS}
Let $\widetilde{\mathfrak{F}}_{CS}$ denote lax CS theory for the double Lie group $\widetilde{G}=G\ltimes \mathfrak{g}^*$, and let $\mathfrak{F}_{BF}$ denote lax BF theory. Then, there exists a map $\widetilde{\mathcal{F}}_{CS} \longrightarrow \mathcal{F}_{BF}$ such that $\widetilde{\cA}\mapsto (\cA,\cB)\equiv\left(c + A + B^\dag + \tau^\dag; \tau + B + A^\dag + c^\dag\right)$ and, denoting 
$$
	f^\bullet_{BF-CS}\coloneqq \frac12 \langle \cB,\cA\rangle,
$$ 
we have
\begin{equation}
	\mathbb{L}_{CS}^\bullet = \mathcal{P}_{f^\bullet_{BF-CS}}(\mathbb{L}_{BF}^\bullet).
\end{equation}
\end{theorem}

\begin{proof}
First of all, we observe that 
$$
	\bD_{CS}^\bullet = (d - \LQ)f^\bullet_{BF-CS}
$$
as it is easily gathered by direct computation of the r.h.s.:
\begin{multline*}
	\frac12 \left( d\cB \cA - \cB d\cA + [\cA,\cB]\cA - d_{\cA}\cB\cA + \cB F_{\cA}\right)
		= \frac12 \left(- [\cA,\cB] + \frac12\cB[\cA,\cA]\right) = -\frac14\cB[\cA,\cA]
\end{multline*}
which coincides with $\bD_{CS}^\bullet[\widetilde{\cA}]$ for $\widetilde{\cA}=\cA + \cB$. Moreover, it is easy to gather that
$$
	L^\bullet_{CS}[\widetilde{\cA}] = L^\bullet_{BF}[\cA,\cB] 
		+ d \left(\frac12 \langle\cB,\cA\rangle\right) = L^\bullet_{BF}[\cA,\cB] + df^\bullet_{BF-CS}.
$$
Then, computing the total Chern--Simons Lagrangian we get
\begin{align*}
		\mathbb{L}^\bullet_{CS} &= L^\bullet_{CS} + \LE\bD^\bullet_{CS} 
			= L^\bullet_{BF} + df^\bullet_{BF-CS} + \LE(d-\LQ)f^\bullet_{BF-CS}\\
		& =L^\bullet_{BF} + df^\bullet_{BF-CS} + \LE df^\bullet_{BF-CS} - \LE\LQ f^\bullet_{BF-CS}\\
		& = L^\bullet_{BF} + df^\bullet_{BF-CS} + \LE df^\bullet_{BF-CS}  
			- \LQ\LE f^\bullet_{BF-CS} - \LQ f^\bullet_{BF-CS}\\
		&= L^\bullet_{BF}+ (d-\LQ) f^\bullet_{BF-CS} + (d - \LQ) \LE f^\bullet_{BF-CS}
\end{align*}
where we used the properties of the Euler vector field of Lemma \ref{QEuler}. Since now $\bD^\bullet_{BF} \equiv 0$ and $\mathbb{L}^\bullet_{BF}=L^\bullet_{BF}$, recalling that
$$
	\mathcal{P}_{f^\bullet_{BF-CS}}(\mathbb{L}_{BF}^\bullet) 
		= \mathbb{L}_{BF}^\bullet - (\LQ- d)(1 + \LE) f^\bullet_{BF-CS},
$$
we can conclude the proof.
\end{proof}

\subsection{BRST type BF Theory}\label{Sect:BRSTBF}
In this concluding section we want to see how the previous discussion can be made analogous to the Chern--Simons case, where the data was put in its BRST-type form (cf. Proposition \ref{CSBRSTtype}).

\begin{proposition}
The $f$-transformed lax BV-BFV theory $\mathcal{P}_{f_{\mathrm{tot}}}(L_{BF}^\bullet, \theta^\bullet_{BF})$, with
$$
	f_{\mathrm{tot}} = BB^\dag + \tau \tau^\dag + \tau B^\dag
$$
is of BRST type. Moreover, the $f$-transformed BV-BFV difference reads
\begin{equation}
	\mathcal{P}_{f_{\mathrm{tot}}}\bD^\bullet_{BF} = BF_A + \tau F_A
\end{equation}
where the classical BF theory is given by $L_{BF}^{cl} = BF_A$.
\end{proposition}

\begin{proof}
Recalling that $\bD^\bullet_{BF} =0$, then 
\begin{multline*}
	\mathcal{P}_{f_{\mathrm{tot}}}\bD^\bullet_{BF} = (d-\LQ) f_{tot} \\
	= - d_A \tau B^\dag - [c,B] B^\dag + B F_A + B[c,B^\dag] - [c,\tau]\tau^\dag + \tau[c,\tau^\dag] \\
	+ \tau d_AB^\dag  + d(\tau B^\dag) - [c,\tau] B^\dag + \tau F_A + \tau[c,B^\dag] = BF_A + \tau F_A.
\end{multline*}
\end{proof}

\section{Yang--Mills theory}\label{YMTheory}
In this section we report a few basic facts about Yang--Mills theory in the BV-BFV formalism. The main reason for this is Remark \ref{YMRemark}, below, which highlights another interpretation and possible application of the BV-BFV differences of Definition \ref{BVBFVDIFFERENCES}, since Yang--Mills theory is not expected to enjoy a particular holographic counterpart on its boundary.

\begin{propdef}
Let $(M,g)$ be a (pseudo-)Riemannian manifold, and let $G$ be a compact, connected, matrix Lie group with Lie algebra $((\mathfrak{g},[\cdot,\cdot])$, endowed with an invariant trace operation. Then, the data 
\begin{equation}
	\mathcal{F}_{YM}\coloneqq T^*[-1]\left(\Omega^1(M,\mathfrak{g}) \oplus \Omega^0(M,\mathfrak{g})[1]\right),
\end{equation}
$L^\bullet_{YM}\in \oloc^{0,\bullet}(\mathcal{F}_{YM})$ and $\theta_{YM}^\bullet\in\oloc^{1,\bullet}(\mathcal{F}_{YM})$ given by, respectively
\begin{subequations}
\begin{align}
	L^\bullet_{YM} & = \mathrm{Tr}\left[ \frac12 F_A \star F_A + A^\dag d_Ac + \frac12c^\dag[c,c]  
		+ cd_A\star F_A + \frac12 A^\dag[c,c] + \frac12[c,c] \star F_A\right]\\
	\theta^\bullet_{YM} & = \mathrm{Tr}\left[A^\dag \delta A + c^\dag \delta c + \delta A \star F_A 
		+ A^\dag \delta c + c \delta(\star F_A)\right],
\end{align}
\end{subequations}
where $\star$ is the Hodge operator defined by $g$, and a vector field $Q\in\xev[1](\mathcal{F}_{YM})$ defined as
\begin{align}
	QA = d_Ac;\ \ Qc=\frac12[c,c];\ \ 
		Q A^\dag = d_A\star F_A + [c,A^\dag]; \ \ Qc^\dag = d_AA^\dag + [c,c^\dag]
\end{align}
defines a lax, strictifiable BV-BFV theory. We will call the data
\begin{equation}
	\mathfrak{F}_{YM}=\left(\mathcal{F}_{YM}, L^\bullet_{YM}, \theta_{YM}^\bullet, Q\right)
\end{equation}
\emph{lax second-order Yang--Mills theory}.
\end{propdef}

\begin{proof}
This is a straightforward computation. We remind the reader that $\delta \star F_A = -\star d_A\delta A$, and that $[F_A,\star F_A]\equiv 0$.
\end{proof}

\begin{remark}
Although admitting a lax BV-BFV description, Yang--Mills theory in $4$ dimensions is generally extendable up to codimension $2$ (cf. with \cite{CMR}).
\end{remark}

\begin{lemma}
The BV-BFV difference for lax second-order Yang--Mills theory reads:
\begin{equation}
	\bD_{YM}^\bullet = \mathrm{Tr}\left[\frac12F_A\star F_A -  d\left(c\star F_A\right) - \frac12[c,c]\star F_A\right] = \mathrm{Tr}\left[\frac12F_A\star F_A + (\LQ -d)(c\star F_A)\right]
\end{equation}
\end{lemma}

\begin{proof}
This is a straightforward computation, since 
\begin{multline*}
	L^\bullet_{YM} - \iota_Q\theta_{YM}^\bullet 
		= \frac12 F_A \star F_A + A^\dag d_Ac + \frac12c^\dag[c,c]  
		+ cd_A\star F_A + \frac12 A^\dag[c,c] + \frac12[c,c] \star F_A\\
	-\left( A^\dag d_Ac + \frac12c^\dag[c,c] + d_Ac\star F_A 
		+ \frac12 A^\dag [c,c] - c[\star F_A,c]\right)\\
	= \frac12F_A\star F_A -  d\left(c\star F_A\right) - \frac12[c,c]\star F_A;
\end{multline*}
however, it is easy to check that
\begin{equation*}
	\LQ (c\star F_A) = \frac12[c,c]\star F_A + c[\star F_A,c] = - \frac12[c,c]\star F_A,
\end{equation*}
completing the proof.
\end{proof}

\begin{remark}\label{YMRemark}
It is worthwhile noticing that the component in codimension $\geq 1$ of the difference $\bD^\bullet$ is $(\LQ-d)$-exact, with the codimension-$1$ component being $d$-exact. On the one hand this is compatible with Theorem \ref{BRSTtypeTheorem}, since lax second-order Yang--Mills theory is manifestly ``of BRST-type''. On the other hand, our choice of presentation is likely relevant for considerations concerning asymptotic symmetries and reconstruction of gauge-invariance of the pre-symplectic potential (here called boundary one-form). As a matter of fact, comparing with \cite[Eq. (2.15)]{DoFr}, we see clearly that the addition to their pre-symplectic potential coincides with
\begin{equation}
	\D^{(1)}=\intl_{\Sigma} [\bD_{YM}^\bullet]^{\mathrm{top} -1} = \intl_{\partial \Sigma} \mathrm{Tr}\left[c\star F_A\right]
\end{equation}
where $\Sigma$ denotes a codimension-$1$ stratum in $M$.
\end{remark}

\begin{remark}
We would like to thank Nicholas J. Teh for pointing out the work of Donnelly and Freidel \cite{DoFr}, a possible relationship with which is discussed in Remark \ref{YMRemark}. A deeper study on how this relates to BV-BFV is currently under investigation by Philippe Mathieu, Nicholas J. Teh and Laura Wells at Notre Dame University and Alexander Schenkel in Nottingham. We refer to their work for further details \cite{ND}. A different branch of this investigation is due to S. Ramirez and N. Teh \cite{RT}.
\end{remark}

An analogous result for Yang--Mills theory in the first-order formalism follows.

\begin{propdef}[\cite{CMR}]
Let Let $(M,g)$ and $G$ be as above. Then the data
\begin{equation}
	\mathcal{F}_{1YM}\coloneqq T^*[-1]\left(\Omega^1(M,\mathfrak{g})\oplus\Omega^{d-2}(M,\mathfrak{g}) \oplus \Omega^0(M,\mathfrak{g})[1]\right),
\end{equation}
$L^\bullet_{1YM}\in \oloc^{0,\bullet}(\mathcal{F}_{1YM})$ and $\theta_{1YM}^\bullet\in\oloc^{1,\bullet}(\mathcal{F}_{1YM})$ given by, respectively
\begin{align}
	L^\bullet_{1YM} & = \mathrm{Tr} \left[BF_A + \frac12 B\star B +  A^\dag d_Ac 
		+ B^\dag [c,B] + \frac12c^\dag[c,c] \right]\\ 
	&+\mathrm{Tr}\left[ B d_Ac + \frac12 A^\dag [c,c]\right] + \mathrm{Tr}\left[\frac12 B[c,c]\right]\\
	\theta^\bullet_{1YM} & 
		= \mathrm{Tr}\left[A^\dag \delta A + B^\dag\delta B + c^\dag \delta c\right] 
			+\mathrm{Tr}\left[ B\delta A + A^\dag \delta c\right ] + \mathrm{Tr}\left[B\delta c\right]
\end{align}
together with a vector field $Q\in\xev[1](\mathcal{F}_{1YM})$ defined as
\begin{align}
	&QA = d_Ac;\ \ QB= [c,B]; \ \ Qc=\frac12[c,c];\\
	&Q A^\dag = d_A\star F_A + [c,A^\dag]; \ \ QB^\dag= F_A + \star B + [c,B^\dag]; \ \ Qc^\dag = d_AA^\dag + [c,c^\dag];
\end{align}
defines a lax BV-BFV theory. We will call the data
\begin{equation}
	\mathfrak{F}_{1YM}=\left(\mathcal{F}_{1YM}, L^\bullet_{1YM}, \theta_{1YM}^\bullet, Q\right)
\end{equation}
\emph{lax first-order Yang--Mills theory}.
\end{propdef}

\begin{lemma}
The BV-BFV difference for lax first order Yang--Mills theory reads:
\begin{equation}
\bD_{1YM}^\bullet = \mathrm{Tr}\left[BF_A + \frac12 B\star B\right].
\end{equation}
\end{lemma}

\begin{proof}
This is easily shown by means of a straightforward computation, or by applying Theorem \ref{BRSTtypeTheorem}, since lax first order Yang--Mills theory is manifestly of BRST type.
\end{proof}

{ 
\section{Poisson Sigma model}\label{Sect:PSM}
Here we discuss a first step towards the application of the method presented in this paper to the Poisson sigma model \cite{Ike,SS} for a Poisson manifold $(M,\Pi)$. This is a fully extended $2$-dimensional theory obtained through the AKSZ construction with target the Hamiltonian manifold $(T^*[1]M, \omega_{std}, \Pi)$ where $\omega_{std}$ is the standard symplectic form and $\Pi$ is interpreted as function on $T^*[1]M$.

\begin{propdef}[\cite{CMR}]
Let $(M,\Pi)$ be a Poisson manifold and $\Sigma$ a two-dimensional manifold. Then, the data 
\begin{equation}
\mathcal{F}_{PSM} \coloneqq \mathrm{Map}(T[1]\Sigma, T^*[1]M)\ni(\bbeta,\BBX),
\end{equation} 
together with $L_{PSM}^\bullet\in \oloc^{0,\bullet}(\mathcal{F}_{PSM})$ and $\theta^\bullet_{PSM}\in \oloc^{1,\bullet}(\mathcal{F}_{PSM})$, given by
\begin{subequations}\begin{align}
	L_{PSM}^{\bullet} &= \langle \bbeta, d\BBX \rangle 
		+ \frac12 \langle\Pi(\BBX),\bbeta \wedge \bbeta \rangle \\
	\theta_{PSM}^{\bullet} &= \bbeta \wedge \delta \BBX 
\end{align}
and with cohomological vector field
\begin{equation}
Q_{PSM}\BBX = d\BBX  + \Pi(\BBX) \bbeta ;\ \ \ Q_{PSM}\bbeta = d\bbeta + \frac12\langle d\Pi(\BBX),\bbeta\wedge \bbeta\rangle 
\end{equation}
\end{subequations}
defines a lax BV-BFV theory. We will call it the \emph{lax Poisson sigma model}.
\end{propdef}

It is well known that the Poisson sigma model does not fit in the BRST setting (see e.g. \cite{CF}), and requires the BV formalism. The main consequence of this fact for the present paper is that there does not exist an $f$-transformation that can turn lax Poisson sigma model into a BRST-type theory (Definition \ref{Def:BRSTTYPE}).

\begin{remark}
It is worthwhile to unpack some of the given expressions in terms of fields of different form degree and ghost number: $\BBX= X + \eta^\dag + \beta^\dag$, and $\bbeta=\beta + \eta + X^\dag$. In particular we will need that in these coordinates the cohomological vector field reads (see e.g. \cite{CF} for more details)
\begin{equation}\label{eq:coordQPSM}
Q_{PSM}\beta=d\Pi \beta \beta; \ \ \ Q_{PSM}\eta^\dag = dX + \eta^\dag d\Pi \beta + \Pi \eta.
\end{equation}
\end{remark}

\begin{lemma}
The BV-BFV difference for lax Poisson $\Sigma$ Model reads:
\begin{equation}
\bD^\bullet_{PSM} = -\frac12 \langle \Pi(\BBX), \bbeta\bbeta\rangle.
\end{equation}
\end{lemma}

\begin{proof}
This is a straightforward computation: 
\begin{align*}
	\bD_{PSM}^\bullet & =  L_{PSM}^\bullet - \iota_{Q_{PSM}}\theta_{PSM}^\bullet 	\\
	& = \langle \bbeta, d\BBX \rangle 
		+ \frac12 \langle\Pi(\BBX),\bbeta \wedge \bbeta \rangle 
		- \bbeta\wedge (d\BBX + \Pi(\BBX)\bbeta) 
		= - \frac12 \langle\Pi(\BBX),\bbeta \wedge \bbeta \rangle.
\end{align*}
\end{proof}

Observe that $\bD_{PSM}^\bullet$ is not necessarily trivial as a $(\LQ -d)$-cocycle, a condition which depends on the characteristics of $\Pi$. As a matter of fact, we have the following.

\begin{proposition}\label{BFPSM}
Let $\Pi(X)$ be linear in $X$ and consider the polarising functional $f_{\text{lin}}= \bbeta \BBX$. Then 
$$\mathcal{P}_{f_{\text{lin}}}\bD_{PSM}^\bullet = 0.$$
\end{proposition}

\begin{proof}
The result is straightforward, and follows from 
$$
	Q(\bbeta \BBX) 
		= d(\bbeta \BBX) + \frac12 \frac{\partial \Pi}{\partial \BBX}\bbeta\bbeta \BBX 
		- \Pi \bbeta\bbeta 
		= d(\bbeta \BBX) - \frac12 - \Pi \bbeta\bbeta
$$
where we used the obvious identity $\frac{\partial \Pi}{\partial \BBX}\BBX = \Pi$, for a linear Poisson structure.
\end{proof}

\begin{remark}
Proposition \ref{BFPSM} is obvious once we realise that the linear Poisson sigma model is symplectomorphic to 2-dimensional BF theory up to a boundary term. The boundary term is exactly $df_{\text{lin}}\equiv d(\bbeta \BBX)$ and one can say that linear Poisson sigma model is 2 dimensional BF theory equipped with a different natural polarisation (in the space of boundary fields).
\end{remark}

We now introduce a new polarising functional that breaks the AKSZ superfield description, in exchange for showing an holographic behaviour at the boundary. The following is related to choosing a polarisation on the boundary where the base variables are $\beta$ and $X$, instead of $\eta^\dag$ and $X$.

\begin{lemma}
Consider the polarising functional $f_{\text{hol}}\coloneqq \beta \eta^\dag$, and let $\{\Sigma^{(k)}\}$ be a stratification of $\Sigma$. The $f$-transformed, transgressed BV-BFV difference in codimension-$1$ reads:
\begin{equation}
	\left[\T\D_{f_{\text{hol}}}^{(1)}\right]_{\mathrm{Map}^0} = \intl_0^1dt \intl_{\Sigma^{(1)}} \langle p(t), dq(t)\rangle.
\end{equation}
with $\left[\T\D_{f_{\text{hol}}}\right] = \intl_{\Sigma^{(1)}} \mathcal{P}_{f_{\mathrm{hol}}}\bD^\bullet_{PSM}$.
\end{lemma}

\begin{proof}
We begin by observing that, using expression \eqref{eq:coordQPSM}, 
$$
	\mathcal{L}_{Q_{PSM}} f_{\text{hol}} = \mathcal{L}_{Q_{PSM}} ( \beta \eta^\dag) = - \beta dX - \beta \Pi \eta. 
$$
Then, the $f$-transformed BV-BFV difference is computed as
$$
	\mathcal{P}_{f_{\text{hol}}} \bD ^\bullet = \bD^\bullet - (\LQ - d ) f_{\text{hol}} = - \frac12 \Pi(\BBX)\bbeta\bbeta + d(\beta \eta^\dag) \beta dX +  \beta\Pi \eta,
$$
and in codimension-$1$ we immediately gather that $\frac12 \Pi(\BBX) \bbeta\bbeta = \beta\Pi \eta$, so that, integrating on a codimension-$1$ stratum $\Sigma^{(1)}$, the only nonvanishing terms are
\begin{equation}
	\D^{(1)}_{f_{\text{hol}}} = \int_{\Sigma^{(1)}} \mathcal{P}_{f_{\text{hol}}} \bD ^\bullet = \int_{\Sigma^{(1)}} \beta dX.
\end{equation}
Now we set up the AKSZ integration construction, i.e. we consider $\mathrm{Map}(T[1]I, \mathcal{F}_{\Sigma^{(1)}})$ where $\mathcal{F}_{\Sigma^{(1)}}$ is the space of codimension-$1$ fields for the Poisson sigma model\footnote{Observe that this is also a mapping space: $\mathcal{F}_{\Sigma^{(1)}}=\mathrm{Map}(T[1]\Sigma^{(1)}, T^*[1]M)$.}. Such maps are parametrised by 
\begin{align*}
	\mathbb{\beta} = \beta (t) + p(t) dt
	= q(t) + x(t) dt
\end{align*}
and the transgression reads
\begin{equation}
	\T\D_{f_{\text{hol}}}^{(1)} = \intl_0^1dt \intl_{\Sigma^{(1)}} \langle p(t), dq(t)\rangle + \langle\beta(t), dx(t)\rangle,
\end{equation}
but since the only maps in degree zero are $p(t)$ and $q(t)$, we immediately get the statement.
\end{proof}

\begin{example}
We conclude this section with a ``toy model" example of the previous construction, when $\Pi=0$. In that case, the transversal Euler--Lagrange equations are $\dot{p}(t) = \dot{q}(t)=0$, which means that the space $\mathrm{dgMap}^0_I(T[1]I, \mathcal{F}_{\Sigma^{(1)}})$ (Definition \ref{ELlocus}) is parametrised by $p = p(0, \theta), q=q(0,\theta)$, for $\theta$ a coordinate on $\Sigma^{(1)}$. Then 
\begin{equation}
	\left[\T\D_{f_{\text{hol}}}^{(1)}\right]_{\mathrm{dgMap}_I^0} = \int_{\Sigma^{(1)}} \langle p, dq\rangle d\theta
\end{equation} 
which recovers \emph{topological} classical mechanics (i.e. zero-Hamiltonian) as the \emph{holographic counterpart} for Poisson sigma model.

We remark that a similar observation appeared in \cite[Section 3, Footnote 3]{CF} in the case of a non-degenerate Poisson structure
%when, to construct a star product by means of a semiclassical expansion of the path integral of Poisson sigma model, the path integration was restricted to trajectories $\gamma\colon \mathbb{R} \to M$: classical paths in a target manifold with zero Hamiltonian. The obvious difference is that such a construction relied on $\Pi$ being nonegenerate 
(hence as far as possible from our toy-example). 
Another comparison extending to the semiclassical case was given in \cite{CMR3}.
It will be interesting to see if and how these examples can be related.
\end{example}

}

\appendix
\section{Proofs of Section \ref{Sec:CS-WZW}}\label{A:ProofsCS-WZW}

\begin{proof}[Proof of Lemma \ref{Lemma:WZWfailure}]
The first statement follows from a standard computation, of which we report only a few steps.
Considering first the classical (i.e. degree-$0$) part, we have 
\begin{align*}
S[A^g] &= \int_M \frac{1}{2}\langle A,dA \rangle + \frac{1}{6}\langle A, [A,A]\rangle -\frac{1}{12}\langle g^{-1}dg, [g^{-1}dg,g^{-1}dg] \rangle \\
&- \frac{1}{2}\langle g^{-1}Ag, dg^{-1}dg \rangle + \frac{1}{2}\langle g^{-1}dg, d(g^{-1}Ag) \rangle  \\ 
&+ \frac{1}{2}\langle g^{-1}Ag, dg^{-1}Ag - g^{-1}Adg\rangle  + \frac{1}{2}\langle g^{-1}dg,[g^{-1}Ag,g^{-1}Ag]\rangle.
\end{align*}
In the first line we find the classical CS action and the WZ functional. The terms in the  second line combine into a total derivative, and yield a boundary term 
$$ \frac{1}{2}\int_{\de M} \langle A,dg g^{-1} \rangle.$$ 
The last line vanishes due to the invariance of the inner product. Finally, turning to the extended BV action we recall that that the covariant derivative of a graded field $\omega$ satisfies $d_A^g\omega^g = (d_A\omega)^g$. It follows immediately from invariance of the inner product that the remaining terms in the extended action \eqref{CSaction} are gauge invariant. The claim follows.

In the case of the polarised action we first compute the effect of a gauge transformation on the polarising functional\footnote{Notice that the $cA^\dag$ part of $f_{min}^{1,0}$ is gauge-invariant and drops out of the calculation.} $f_{min}^{1,0}$:
\begin{align*} 
	\intl_{\de M} f_{min}^{1,0}[\calA^g] - f_{min}^{1,0}[\calA] 
		& = \frac{1}{2} \int_{\de M} \Big\{ \langle g^{-1}A^{1,0}g, g^{-1}A^{0,1}g\rangle 
			+  \langle g^{-1}\de g, g^{-1}A^{0,1}g\rangle \\
		 & + \langle g^{-1}A^{1,0}g, g^{-1}\bar\de g \rangle + \langle g^{-1} \de g, g^{-1}
			\bar\de g\rangle  -  \langle A^{1,0},A^{0,1}\rangle \Big\}\\ 
		& =   \frac{1}{2}\int_{\de M} \langle g^{-1}\de g, g^{-1}A^{0,1}g\rangle 
			+ \langle g^{-1}A^{1,0}g, g^{-1}\bar\de g \rangle + 
			\langle g^{-1} \de g, g^{-1}\bar\de g\rangle.
\end{align*} 
Then,
\begin{multline*}
	S^{1,0}[\calA^g]-S^{1,0}[\calA]  = S[A^g]- S[A] +  
			\intl_{\de M} f_{min}^{1,0}[A^g] - f_{min}^{1,0}[A] \\
		= \int_{\de M} \frac{1}{2}\langle g^{-1}Ag, g^{-1}dg \rangle - \int_M \frac{1}{12}
			\langle g^{-1}dg,[g^{-1}dg,g^{-1}dg]\rangle  \\ 
		+ \frac{1}{2}\int_{\de M} \langle g^{-1}\de g, g^{-1}A^{0,1}g\rangle + 
			\langle g^{-1}A^{1,0}g, g^{-1}\bar\de g \rangle + 
			\langle g^{-1} \de g, g^{-1}\bar\de g\rangle  \\ 
		= \int_{\de M} \langle g^{-1}A^{1,0}g, g^{-1}\bar\de g \rangle + 
			\frac{1}{2} \langle g^{-1} \de g, g^{-1}\bar\de g\rangle 
			-\int_M \frac{1}{12}\langle g^{-1}dg,[g^{-1}dg,g^{-1}dg]\rangle.
\end{multline*}

\end{proof}

\begin{proof}[Proof of Lemma \ref{Lemma:gWZWinvariance}]
This follows immediately from 
\begin{multline}\label{S_gWZ gt property}
	S_{gWZ}(h^{-1}g,A^h) = S_{CS}(A^g) - S_{CS}(A^h) =\\
	= \Big(S_{CS}(^g A)-S_{CS}(A)\Big) - \Big(S_{CS}(^h A)
		+S_{CS}(A)\Big) = S_{gWZ}(g,A) - S_{gWZ}(h,A).
\end{multline}
\end{proof}

\begin{proof}[Proof of Lemma \ref{Lemma:ELtransverse}]
Using the defining property of the path-ordered exponential, $\frac{d}{dt}\mathrm{Pexp}(\int_0^t \gamma_s ds)= \mathrm{Pexp}(\int_0^t \gamma_s ds) \gamma_t$, we have that $g^{-1}_t\dot{g}_t = \gamma_t$. Hence,
\begin{multline*}
	\frac{d}{dt} A^{g_t} = \frac{d}{dt} (g^{-1}_tA\,g_t+ g^{-1}_tdg_t) = [g^{-1}_tA\,g_t, \gamma_t] 
		- \gamma_t g^{-1}_tdg_t + g^{-1}_td\dot{g}_t \\
	= [g^{-1}_tA\,g_t, \gamma_t] + [g^{-1}_tdg_t,\gamma_t] + d\gamma_t = d_{A^{g_t}}\gamma_t.
\end{multline*}

The second claim follows from a simple direct calculation: denoting $\phi_t\equiv g_t^{-1}dg_t$
\begin{multline*}
	\dot{\phi} = \frac{d}{dt}g_t^{-1} dg_t + g^{-1}_t d (\dot{g}_t) = -g_t^{-1}\dot{g}_t g_t^{-1} dg_t 
		+ g_t^{-1}d (\dot{g}_t) \\
	= - \gamma_t g_t^{-1} dg_t + g_t^{-1} dg_t \gamma_t + d\gamma_t = d\gamma_t 
		+ [g_t^{-1} dg_t,\gamma_t]= d_{\phi_t}\gamma_t.
\end{multline*}
\end{proof}

\begin{proof}[Proof of Lemma \ref{Lemma:threedterm}]
The Wess--Zumino functional in Equation \eqref{WZ} does not depend on the extension $\widetilde{g}$: choosing a different extension changes $S_{WZ}$ by a constant. In particular this is irrelevant when taking a time derivative. Hence, let us choose an extension $\widetilde{g}_t\coloneqq \mathrm{Pexp}(\int_0^t \widetilde{\gamma}_s ds)$, with $\widetilde{\gamma}_t\colon M \to \mathfrak{g}$ an extension of $\gamma_t$, i.e. $\widetilde{\gamma}_t\vert_{\partial M}=\gamma_t$. For simplicity of notation we drop the tildes in what follows. Let us denote again $\phi_t\equiv g_t^{-1}dg_t$. Because $\phi_t$ is the (pullback of the) Maurer--Cartan form on $G$, in addition to Lemma \ref{Lemma:ELtransverse} we have that 
\begin{equation*}
	d\phi_t = -\frac12[\phi_t,\phi_t].
\end{equation*}
Then, we can directly compute
\begin{multline*}
	\frac{d}{dt}S_{WZ}[g_t]=-\frac{d}{dt} \int_M \frac{1}{12}\langle \phi_t,[\phi_t,\phi_t]\rangle 
		= \frac{d}{dt} \int_M \frac{1}{6}\langle \phi_t,d\phi_t\rangle 
		= \frac{1}{6}\int_M \langle \dot{\phi}_t,d\phi_t\rangle + \langle {\phi}_t,d\dot{\phi}_t\rangle \\
	= \frac{1}{6}\intl_M  \langle d\gamma_t, d\phi_t\rangle + \langle[\phi_t,\gamma_t],d\phi_t\rangle 
		+ \langle\phi_t,d(d\gamma_t + [\phi_t,\gamma_t])\rangle \\
	= \frac{1}{6}\intl_M \langle d\gamma_t, d\phi_t\rangle - \langle[\phi_t,\phi_t],d\gamma_t\rangle 
		= \frac{1}{2}\intl_M \langle d\gamma_t, d\phi_t\rangle 
		= -\frac{1}{2}\intl_M d \left[\langle d\gamma_t, \phi_t\rangle\right] 
		= \frac{1}{2}\intl_{\partial M} \langle \phi_t, d\gamma_t\rangle.
\end{multline*}
\end{proof}

\begin{proof}[Proof of Proposition \ref{Proposition:WZderivatives}]
Using Lemma \ref{Lemma:ELtransverse}, Lemma \ref{Lemma:threedterm} and denoting again $\phi_t\equiv g_t^{-1}dg_t$, we compute 
\begin{multline*}
	\frac{d}{dt} S_{gWZ} = \frac12\int_{\de M} \langle\frac{d}{dt}\left(g_t^{-1}A\,g_t\right), \phi_t\rangle 
		+  \langle g_t^{-1}A\,g_t, \dot{\phi}_t\rangle 
		- \frac{d}{dt}\int_M \frac{1}{12}\langle \phi_t,[\phi_t,\phi_t]\rangle \\
	= \frac12 \int_{\de M}\langle -\gamma_t\,g_t^{-1}A\,g_t, \phi_t\rangle 
		+ \langle g_t^{-1}A\,g_t\,\gamma_t, \phi_t\rangle + \langle g_t^{-1}A\,g_t,d\gamma_t\rangle 
		+ \langle g_t^{-1}A\,g_t,[\phi_t,\gamma_t]\rangle + \langle \phi_t, d\gamma_t \rangle\\
	= \frac12 \int_{\de M}\langle[g_t^{-1}A\,g_t,\gamma_t], \phi_t\rangle + \langle g_t^{-1}A\,g_t,[\phi_t,\gamma_t] \rangle + \langle\left(g_t^{-1}A\,g_t + \phi_t\right), d\gamma_t\rangle =  \frac12 \int_{\de M} \langle A^{g_t}, d\gamma_t\rangle 
\end{multline*}
where we used $\langle g_t^{-1}A\,g_t,[\phi_t,\gamma_t]\rangle = - \langle [g_t^{-1}A\,g_t,\gamma_t],\phi_t\rangle$.

The details of the calculation for $S_{gWZW}^{1,0}$ is identical, upon replacing $g^{-1}dg$ with $g^{-1}\bar{\partial}g$, the connection $A$ with $A^{1,0}$, and expanding $d=\partial + \bar{\partial}$ in the right-hand side of formula \eqref{phithreetermvariation}. 
\end{proof}


\begin{thebibliography}{99}
\bibitem[ABM]{ABM} A. Alekseev, Y. Barmaz and P. Mnev, \emph{Chern-Simons Theory with Wilson Lines and Boundary in the BV-BFV Formalism}, Journal of Geometry and Physics {\bf 67}, (2013), 1-15.
\bibitem[AKSZ]{AKSZ} M. Alexandrov, M. Kontsevich, A. Schwarz and O. Zaboronsky, \emph{The geometry of the master equation and topological quantum field theory}, Int. J. Mod. Phys. {\bf A12}, 1405-1430 (1997).
\bibitem[ANXZ]{Alekseev} A. Alekseev, F. Naef, X. Xu and C. Zhu, \emph{Chern--Simons, Wess--Zumino and other cocycles from Kashiwara Vergne Associators}, Lett. Math. Phys. {\bf 108}(3), 757-778 (2018).
\bibitem[And]{Anderson} I. M. Anderson, \emph{The variational bicomplex}, unfinished book, available at http://math.uni.lu/?michel/data
\bibitem[AdPW]{ADPW} S. Axelrod, S. Della Pietra and E. Witten, \emph{Geometric quantization of Chern-Simons theory}, J. Diff. Geom. {\bf 33}, 787--902, (1991).
\bibitem[BBH]{BBH} { G. Barnich, F. Brandt and M. Henneaux, \emph{Local BRST cohomology in gauge theories}, Phys. Rept. {\bf 338}, 439 (2000).}
\bibitem[BV77]{BV77} I. A. Batalin and G. A. Vilkovisky, \emph{Relativistic S-matrix of dynamical systems with boson and fermion costraints}, Phys. Lett. {\bf B 69}(3), 309-312 (1977).
\bibitem[BV81]{BV81} I. A. Batalin and G. A. Vilkovisky. {\it Gauge algebra and quantization}, Phys. Lett. {\bf B 102}(1), 27-31 (1981).
\bibitem[BF83]{BF83} I. A. Batalin and E. S. Fradkin, \emph{A generalized canonical formalism and quantization of reducible gauge theories}, Phys. Lett. {\bf B 122}(2), 157-164 (1983).
\bibitem[BRS1]{BRS1} C. Becchi, A. Rouet and R. Stora, \emph{The abelian Higgs Kibble model, unitarity of the S-operator}, Phys. Lett. {\bf B 52} (1974) 344.
\bibitem[BRS2]{BRS2} C. Becchi, A. Rouet and R. Stora, \emph{Renormalization of the abelian Higgs-Kibble model}, Commun. Math. Phys. {\bf 42} (1975) 127.
\bibitem[BRS3]{BRS3} C. Becchi, A. Rouet and R. Stora, \emph{Renormalization of gauge theories}, Ann. Phys. {\bf 98}(2), (1976) pp. 287-321.
\bibitem[BF]{FB} S. Bieri, J. Fr\"ohlich, \emph{Physical principles underlying the quantum Hall effect}, Comptes Rendus Physique 12:332-346, (2011).

\bibitem[BFR]{BFR} {  R. Brunetti, K. Fredenhagen and K. Rejzner, \emph{Quantum Gravity from the Point of View of Locally Covariant Quantum Field Theory}, Commun. Math. Phys. 345 (3) 741-779 (2016).}
\bibitem[CaS]{CaS} G. Canepa and M. Schiavina, \emph{Fully extended BV-BFV description of general relativity in three dimensions}, arXiv:1905.09333 [math-ph].
\bibitem[Car]{Carlip} S. Carlip, \emph{Conformal field theory, (2+1)-dimensional gravity, and the BTZ black hole}, Class. Quant. Grav., {\bf 22} (2005) 85-124.
\bibitem[CCFM]{CCFM} A. S. Cattaneo, P. Cotta-Ramusino, J. Fr\"ohlich and M. Martellini, {
Topological BF theories in 3 and 4 dimensions}, Journal of Mathematical Physics {\bf 36}, 6137 (1995).

\bibitem[CF]{CF} {  A. S. Cattaneo and G. Felder, \emph{A path integral approach to the Kontsevich quantization formula}, }
\bibitem[CMW]{CMW} A. S. Cattaneo, P. Mnev, K. Wernli, \emph{Split Chern--Simons theory in the BV-BFV formalism}, in ``Quantization, Geometry and Noncommutative Structures in Mathematics and Physics'' (pp. 293-324). Springer, Cham.
\bibitem[CMR14]{CMR} A. S. Cattaneo, P. Mnev and N. Reshetikhin, \emph{Classical BV theories on manifolds with boundary}, Comm. Math. Phys. {\bf 332}(2), 535-603 (2014).
\bibitem[CMR18]{CMR2} A. S. Cattaneo, P. Mnev and N. Reshetikhin, \emph{Perturbative Quantum Gauge Theories on Manifolds with Boundary}, Comm. Math. Phys {\bf 357}(2), 631-730 (2018).
\bibitem[CMR18a]{CMR3} A. S. Cattaneo, P. Mnev and N. Reshetikhin, \emph{Poisson Sigma Model and Semiclassical Quantization of Integrable Systems}, in Ludwig Faddeev Memorial Volume, World Scientific, pp. 93-118 (2018).

\bibitem[CR]{CR} A. S. Cattaneo and C. Rossi, \emph{Higher-dimensional BF theories in the Batalin-Vilkovisky formalism: The BV action and generalized Wilson loops}, Commun. Math. Phys. {\bf 221} (2001) 591-657.

\bibitem[CS15]{CSEH} A. S. Cattaneo and M. Schiavina, \emph{BV-BFV approach to general relativity: Einstein--Hilbert action}, J. Math. Phys. {\bf 57}, 023515 (2016).
\bibitem[CS16]{CStime} A. S. Cattaneo and M. Schiavina, \emph{On time}, Lett. Math. Phys. {\bf107}(2), pp 375-408, (2017).
\bibitem[CS17]{CSPCH} A. S. Cattaneo and M. Schiavina, \emph{BV-BFV approach to general relativity: Palatini--Cartan--Holst action}, arXiv:1707.05351, to appear in Advances in Theoretical and Mathematical Physics.
{ 
\bibitem[CSS]{CSS} A. S. Cattaneo, M. Schiavina and I. Selliah, {\emph BV-equivalence between triadic gravity and BF theory in three dimensions}, Lett. Math. Phys. 108 (8) (2018) 1873-1884.
\bibitem[CM]{CM} G. Y. Cho and J. Moore, \emph{Topological BF field theory description of topological insulators}, Annals Phys. {\bf 326}, 1515-1535, (2011).
\bibitem[CSW]{CSW} I. Contreras, M. Schiavina and K. Wernli, \emph{In progress}.

\bibitem[Co]{Co} K. Costello, \emph{Renormalization and Effective Field Theory}, Mathematical Surveys and Monographs, Volume 170, Americal Mathematical Society (2011).
\bibitem[CG]{CG} K. Costello and O. Gwilliam, \emph{Factorization Algebras in Quantum Field Theory, Volume 1}, New Mathematical Monographs, Volume 31, Cambridge University Press (2016).}
\bibitem[CHvD]{CHvD} O. Coussaert, M. Henneaux and P. van Driel, \emph{The asymptotic dynamics of three-dimensional Einstein gravity with a negative cosmological constant}, Class.Quant.Grav. {\bf 12}, 2961-2966, (1995).

\bibitem[Del]{Delgado} N. L. Delgado, \emph{Lagrangian field theories: ind/pro-approach and $L_\infty$-algebra of local observables}, PhD thesis, 2017.
\bibitem[DeFr]{DF} P. Deligne and D.S. Freed, \emph{Classical Field Theory. Quantum Fields and Strings: A course for Mathematicians} Ed. by D. S. Freed, L. C. Jeffrey, D. Kazhdan, J. W. Morgan, D. R. Morrison, P. Deligne, P. Etingof and E. Witten. American Mathematical Society, Providence, Rodhe Island, volume 1, pp. 137-226, 1999

\bibitem[DoFr]{DoFr} W. Donnelly and L. Freidel, \emph{Local subsystems in gauge theory and gravity}, Journal of High Energy Physics, {\bf 09} (2016) 102.

\bibitem[EMSS]{EMSS} S. Elitzur, G. Moore, A. Schwimmer and N. Seiberg, \emph{Remarks on the canonical quantization of the Chern-Simons-Witten theory}, Nuclear Physics {\bf B 326}(1), 108-134, (1989).

\bibitem[F]{F1} J. Fr\"ohlich, \emph{Chiral Anomaly, Topological Field Theory, and Novel States of Matter}, to appear in: "Ludwig Faddeev Memorial Volume: A Life in Mathematical Physics", edited by Molin Ge, Antti Niemi, Kok Khoo Phua, and Leon A Takhtajan (World Scientific, 2018)

\bibitem[GK]{GK} K. Gawedzki and A. Kupiainen, \emph{$SU(2)$ Chern-Simons theory at genus zero}, 1991
\bibitem[Gei]{Gei} M. Geiller, \emph{Edge modes and corner ambiguities in3d Chern--Simons theory and gravity}, Nucl. Phys. {\bf B924}, 312-365, (2017).

\bibitem[Ike]{Ike} { N. Ikeda, \emph{Two-dimensional gravity and nonlinear gauge theory}, Ann. Phys. 235(1994) 435- 464.}
\bibitem[IM]{IM} R. Iraso and P. Mnev, \emph{Two-Dimensional Yang--Mills Theory on Surfaces with Corners in Batalin--Vilkovisky Formalism}, Comm. Math. Phys (2019).

\bibitem[MSZ]{MSZ} { J. Ma\~nes, R. Stora and B. Zumino, \emph{Algebraic study of chiral anomalies}, Commun. Math. Phys 102 (1985) 157-174.}
\bibitem[MSTW]{ND} P. Mathieu, A. Schenkel, N. J. Teh and L. Wells, \emph{Homological perspective on edge modes in linear Yang-Mills theory}, arXiv:1907.10651 (2019).

\bibitem[Mn]{Mn1} P. Mnev, \emph{Discrete BF theory}, arXiv:0809.1160.

\bibitem[Re]{Re} { K. Rejzner, \emph{Perturbative Algebraic Quantum Field Theory}, Mathematical Physics Studies, Springer International Publishing (2016).
\bibitem[RT]{RT} S. Ramirez and N. J. Teh, \emph{Abandoning Galileo's Ship: The quest for non-relational empirical significance}, in progress.
\bibitem[SS]{SS} P. Schaller and T. Strobl, Poisson structure induced (topological) field theories, Modern Phys. Lett. A 9 (1994), no. 33, 3129-3136}
\bibitem[Ty]{T} I.V. Tyutin, \emph{Gauge Invariance in Field Theory and Statistical Physics in Operator Formalism}, Lebedev Physics Institute preprint 39 (1975), arXiv:0812.0580.

\bibitem[WZ]{WZ} J. Wess and B. Zumino, \emph{Consequences of anomalous Ward identities}, Phys. Lett. 37B, 95 (1971).
\bibitem[W83]{W83} E. Witten, \emph{Global aspects of current algebra}, Nucl. Phys. B223, 422 (1983).
\bibitem[W84]{W84} E. Witten, \emph{Non-Abelian bosonization in two dimensions}, Commun. Math. Phys. 92, 455 (1984).
\bibitem[W86]{W1} E. Witten, \emph{Topological Sigma Models}, Comm. Math. Phys {\bf{118}}(3), 411-449 (1988)
\bibitem[W89]{W89} E. Witten, \emph{Quantum field theory and the Jones polynomial}, Comm. Math. Phys. {\bf 121}(3), (1989), 351-399.

\bibitem[Z85]{Z1} { B. Zumino, \emph{Cohomology of gauge groups: cocycles and Schwinger terms}, Nucl. Phys. {\bf B} 253 (1985) 477-493.}







\end{thebibliography}
\end{document}